\documentclass{amsart}
             
\usepackage{amssymb}
\usepackage{amsthm}                
\usepackage{amstext}
\usepackage{amsbsy}
\usepackage{enumerate}
\usepackage[noadjust]{cite}
\usepackage{hyperref}
\usepackage{mathtools}
\usepackage[usenames,dvipsnames]{color}
\usepackage{appendix}

\usepackage{tikz}
\usetikzlibrary{matrix}

\makeatletter
\renewcommand*\env@matrix[1][\arraystretch]{%
  \edef\arraystretch{#1}%
  \hskip -\arraycolsep
  \let\@ifnextchar\new@ifnextchar
  \array{*\c@MaxMatrixCols c}}
\makeatother

\newtheorem{thm}{Theorem}[section]

\newtheorem{lem}[thm]{Lemma}
\newtheorem{prop}[thm]{Proposition}

\theoremstyle{remark}
\newtheorem{rem}[thm]{Remark}

\theoremstyle{definition}
\newtheorem{defn}[thm]{Definition}

\theoremstyle{plain}

\newcommand{\norm}[1]{\left\Vert#1\right\Vert}
\newcommand{\abs}[1]{\left\vert#1\right\vert}
\newcommand{\set}[1]{\left\{#1\right\}}
\newcommand{\R}{\mathbb{R}}
\newcommand{\C}{\mathbb{C}}
\newcommand{\N}{\mathbb{N}}
\newcommand{\Z}{\mathbb{Z}}

\newcommand{\diff}{\mathrm{Diff}}
\newcommand{\emb}{\mathrm{Emb}}

\renewcommand{\dim}{\mathrm{dim}}
\newcommand{\loc}{\mathrm{loc}}
\newcommand{\Tr}{\mathrm{Tr}}
\newcommand{\hdim}{\mathbf{dim}_\mathrm{H}}
\newcommand{\lhdim}{\mathbf{dim}_\mathrm{H}^\mathrm{loc}}
\newcommand{\hdist}{\mathbf{dist}_\mathrm{H}}
\newcommand{\diam}{\mathbf{diam}}
\newcommand{\bdim}{\mathbf{dim}_\mathrm{B}}

\newcommand{\GL}{\mathrm{GL}}
\newcommand{\singlebar}{\hspace{1mm}|\hspace{1mm}}
\begin{document}

\title[Classical and quantum Ising models]{Properties of 1D classical and quantum Ising models: rigorous results}

\author[W. N. Yessen]{William N. Yessen}
\email{wyessen@math.uci.edu}
\address{Department of Mathematics, UC Irvine, Irvine, CA 92697}
\thanks{The author was supported by the NSF grants DMS-0901627, PI: A. Gorodetski and IIS-1018433, PI: M. Welling and Co-PI: A. Gorodetski}

\date{\today}

\begin{abstract}

In this paper we consider one-dimensional classical and quantum spin-$1/2$ quasi-periodic Ising chains, with two-valued nearest neighbor interaction modulated by a Fibonacci substitution sequence on two letters. In the quantum case, we investigate the energy spectrum of the Ising Hamiltonian, in presence of constant transverse magnetic field. In the classical case, we investigate and prove analyticity of the free energy function when the magnetic field, together with interaction strength couplings, is modulated by the same Fibonacci substitution (thus proving absence of phase transitions of any order at finite temperature). We also investigate the distribution of Lee-Yang zeros of the partition function in the complex magnetic field regime, and prove its Cantor set structure (together with some additional qualitative properties), thus providing a rigorous justification for the observations in some previous works. In both, quantum and classical models, we concentrate on the ferromagnetic class.

\end{abstract}

\maketitle

\section{Introduction}\label{sec:intro}

Since the discovery of quasicrystals \cite{Levine1984,Levine1986, Shechtman1984, Socolar1986}, quasi-periodic models in mathematical physics have formed an active area of research. The method of trace maps, originally introduced in \cite{Kadanoff0000, Kohmoto1983, Ostlund1983} (see also \cite{Horowitz1972, Southcott1979, Traina1980, Jorgensen1982, Kohmoto1992, Baake1999, Roberts1994b, Roberts1994} and references therein), has provided a means for rigorous investigation into the physical properties of one-dimensional quasi-periodic structures, leading, for example, to fundamental results in spectral theory of discrete Schr\"odinger operators and Ising models on one-dimensional quasi-periodic lattices (for Schr\"odinger operators: \cite{Kohmoto1983, Casdagli1986, Suto1987, Damanik2009, Damanik2010, Damanik2000, Damanik2008, Damanik2005, Bellissard1989, Raymond1997, Simone2009}, for Ising models: \cite{Tsunetsugu1987, Benza1989, Ceccatto1989, Hermisson1997, Doria1988, Benza1990, You1990, Tong1997}, and 
references therein).

Quasiperiodicity and the associated trace map formalism is still an area of active investigation, mostly in connection with their applications in physics. In this paper we use these techniques to investigate ergodic families of classical and quantum quasi-periodic one-dimensional Ising models.

\subsection{One-dimensional quantum Ising chains}\label{subsec:intro-quantum-chains}

One-dimensional quasi-periodic quantum Ising spin-$1/2$ chains have been investigated (analytically and numerically) over the past two decades \cite{Benza1989, Ceccatto1989, Hermisson1997, Baake1999,Doria1988, Benza1990, You1990, Igloi1988, Turban1994,Igloi1997, Igloi1998, Igloi2007}. Numerical and some analytic results suggest Cantor structure of the energy spectrum, with nonuniform local scaling (i.e. a multifractal). The multifractal structure and fractal dimensions of the energy spectrum have only recently been rigorously investigated in \cite{Yessen2011}, and only in some special cases. In this paper we extend the results of \cite{Yessen2011} to cover the general case.

In the case of quantum Ising models, we investigate the energy spectrum of the Ising spin-$1/2$ chains with two-valued nearest neighbor couplings arranged in a quasi-periodic sequence, with uniform, transverse magnetic field. We shall concentrate on the quasi-periodic sequence generated by Fibonacci substitution on two symbols.

By mapping the single-fermion Ising Hamiltonian canonically to a Jacobi matrix and employing the techniques that we developed in \cite{Yessen2011a}, we  rigorously demonstrate the multifractal structure of the energy spectrum (thus extending the results of \cite{Yessen2011}--that is, answering \cite[Conjecture 2.2]{Yessen2011}), and investigate its topological, measure-theoretic and fractal-dimensional properties. This provides analytic justification for the results that were observed through numeric experiments (and, in some cases, rather soft analysis) by some authors -- we note, in particular, \cite{Benza1989, Ceccatto1989, Benza1990, You1990, Doria1988}.

In fact, we consider a family of Ising models where interaction strength couplings are modulated by sequences from the hull of the Fibonacci substitution sequence, and by employing the results of \cite{Yessen2011a}, as mentioned above, we show that the spectrum is independent of the choice of a sequence in the hull. 

After investigating the energy spectrum of a single fermion, we pass to topological description of the energy spectrum for many noninteracting fermions. In this case we show that for sufficiently many fermions, the spectrum is an interval whose length grows linearly as the number of particles is increased.

\subsection{One-dimensional classical Ising chains}\label{subsec:intro-classical-chains}

The classical Ising spin chain in one dimension with quasi-periodic interaction (nearest neighbor) and magnetic field have also been investigated both numerically and analytically (see, for example, \cite{Tracy1988, Tracy1988a, Baake1999, Baake1995, Grimm1990, Grimm2002, Grimm1995, Hermisson1997, Tong1997, Tsunetsugu1987, Barata2001}, and references therein). A particularly curious problem is the analyticity of the free energy function in the case the interaction strength couplings as well as the external field are modulated by a quasi-periodic sequence (Fibonacci, say). Since the corresponding transfer matrices do not in general commute, this becomes a relatively nontrivial problem in comparison with the periodic case (we comment in more detail on this below, or see \cite{Baake1999}). By employing the trace map as an analytic map on the three-dimensional complex manifold, we show that the free energy function is in fact analytic in the thermodynamic limit.

In fact, we consider a family of classical one-dimensional Ising models where interaction strength couplings and magnetic field are modulated by a quasi-periodic sequence from the hull of the Fibonacci substitution sequence. It is well known that the partition function can be computed as the trace of the associated transfer matrices. These transfer matrices form a $\GL(2,\R)$ cocycle over the shift map on the hull. We show that in the thermodynamic limit, the free energy is analytic, is independent of the choice of a sequence in the hull, and in fact coincides with the Lyapunov exponent of the associated $\GL(2, \R)$ cocycle. 

Another problem is to describe topology of the Lee-Yang zeros of the partition function in the thermodynamic limit, when the former is viewed as a function of complex magnetic field. By the famous theorem of T. D. Lee and C. N. Yang \cite{Lee1952}, these zeros lie on $S^1\subset \C$ when the couplings are ferromagnetic. Numerical experiments as well as some analysis in \cite{Baake1995} suggest, among other things, that the distribution of zeros in the thermodynamic limit forms a Cantor set that exhibits gap structure similar to those that were investigated by J. Bellissard et. al. in \cite{Bellissard1992} (see also \cite{Bellissard1989, Bellissard1990, Bellissard1991} and \cite{Raymond1997}), when the sequence of interaction strength couplings is modulated by the Fibonacci substitution. This is also suggested by some results of J. C. A. Barata and P. S. Goldbaum in \cite{Barata2001}. Both aforementioned investigations, however, fall short of proving the postulated gap structure or fractal properties of the 
Lee-Yang zeros (albeit providing rather convincing heuristic arguments, in part based on numeric computation).

We prove that the Lee-Yang zeros in the thermodynamic limit do indeed form a Cantor set. In fact, we show that topological, measure-theoretic and fractal-dimensional properties of the set of Lee-Yang zeros on the circle are the same as those of the single-fermion energy spectra of quantum Ising quasicrystals, which are the same as those of the spectra of recently considered quasi-periodic Jacobi matrices.

A natural question remains: are the Lee-Yang zeros in the thermodynamic limit independent of the choice of a sequence in the hull? At present we are unable to answer this definitely, but we conjecture a positive answer.

\section{The Fibonacci substitution sequence, symbolic dynamics on the hull and Lyapunov exponents of associated \texorpdfstring{$\GL(k,\R)$}{} cocycles}\label{sec:prelim}

Let us begin by constructing the Fibonacci substitution sequence, the hull of this sequence, and the associated symbolic dynamics (a left shift map on the hull). These constructions will play a crucial role in the rest of our investigation (to each sequence in the hull we shall associate an Ising model with interaction strength couplings and magnetic field modulated by the chosen sequence). We then investigate Lyapunov exponents of $\GL(k,\R)$ cocycles over the shift map on the hull (this will play a role in our investigation of the free energy function in the thermodynamic limit of classical quasi-periodic Ising models).

\subsection{The Fibonacci substitution sequence and symbolic dynamics on the hull}\label{subsec:fibonacci-substitution}

Let $\mathcal{A} = \set{a,b}$; $\mathcal{A}^*$ denotes the set of finite words over $\mathcal{A}$. The Fibonacci substitution $S: \mathcal{A}\rightarrow\mathcal{A}^*$ is defined by $S: a\mapsto ab$, $S: b\mapsto a$. We formally extend the map $S$ to $\mathcal{A}^*$ and $\mathcal{A}^\N$ by
\begin{align*}
S: \alpha_1\alpha_2\cdots\alpha_k\mapsto S(\alpha_1)S(\alpha_2)\cdots S(\alpha_k)
\end{align*}
and
\begin{align*}
S: \alpha_1\alpha_2\cdots\mapsto S(\alpha_1)S(\alpha_2)\cdots.
\end{align*}
There exists a unique \textit{substitution sequence} $u\in\mathcal{A}^\N$ with the following properties \cite{Queffelec2010}:
\begin{align}\label{eq_u_prop}
\text{\hspace{2mm}} u_1\cdots u_{F_k} = S^{k-1}(a);
\text{\hspace{2mm}} S(u) = u;
\text{\hspace{2mm}} u_1\cdots u_{F_{k+2}} = u_1\cdots u_{F_{k + 1}}u_1\cdots u_{F_k},
\end{align}
where $\set{F_k}_{k\in\N}$ is the sequence of Fibonacci numbers: $F_0 = F_1 = 1;\text{\hspace{2mm}}F_{k\geq 2} = F_{k-1} + F_{k-2}$. In fact, this sequence can be obtained constructively via the following recursive procedure. Starting with $a$, we apply the substitution $S$ recursively:
\begin{align*}
a\overset{S}{\longmapsto} ab\overset{S}{\longmapsto} aba\overset{S}{\longmapsto} abaab\overset{S}{\longmapsto} abaababa\overset{S}{\longmapsto}\cdots.
\end{align*}
From now on we reserve the notation $u$ for this specific sequence. 

Let us remark that the sequence $u$ can also be obtained as a sampling of an irrational circle rotation on two intervals as follows. Let $\phi = (\sqrt{5}+1)/2$ denote the golden mean. Define the sequence $u = \set{u_n}_{n\in\N}$ by
\begin{align*}
u_n = a\chi_{[1-\phi, 1)}(n\phi\hspace{1mm}\mathrm{mod}\hspace{1mm}1) + b\chi_{[1-\phi, 1)^\mathrm{c}}(n\phi\hspace{1mm}\mathrm{mod}\hspace{1mm}1),
\end{align*}
where $\bullet^c$ denotes the set complement.

Observe that for any Fibonacci number $F_k$, $k\geq 2$, the finite word $[u]_1\cdots [u]_{F_k}$ contains $F_{k-1}$ symbols $a$ and $F_{k-2}$ symbols $b$ (where by $[u]_i$ we mean the $i$th element of the sequence $u$).

So far $u$ is a sequence over $\N$. Let us extend $u$ arbitrarily to the left, and call the resulting sequence over $\Z$, $\widehat{u}$. Endow $\mathcal{A}$ with discrete topology, and $\mathcal{A}^\Z$ with the corresponding product topology, and let $T: \mathcal{A}^\Z\hookleftarrow$ denote the left shift. Define the \textit{hull} of $\widehat{u}$, $\Omega$, as follows.
\begin{align*}
\Omega = \set{\omega\in\mathcal{A}^\Z: \lim_{i\rightarrow\infty}T^{n_i}(\widehat{u}) = \omega; \text{ where } \set{n_i}\subset\N, n_i\uparrow\infty}.
\end{align*}
Note that there exists a sequence $w_s\in \Omega$ such that $\set{[w_s]_n}_{n\geq 1}$ coincides with the sequence $u$; indeed, any limit point of the set $\set{T^{F_k}(\widehat{u}): k\in\N}$, with $F_k$ the $k$th Fibonacci number, would serve as such a sequence. From now on we fix such a sequence and reserve the notion $w_s$ for it.

It is easy to see that $\Omega\subset\mathcal{A}^\Z$ is compact and invariant under $T$. Moreover, it is known that $T|_\Omega$ is strictly ergodic (that is, $T|_\Omega$ is topologically minimal and uniquely ergodic)---minimality follows from Gottschalk's theorem \cite{Gottschalk1963} and unique ergodicity from Oxtoby's ergodic theorem \cite{Oxtoby1952} and repetitive nature of $u$; for a comprehensive exposition see \cite[Chapter 5]{Queffelec2010}. Let us reserve the notation $\mu$ for the unique ergodic measure for $T$ on $\Omega$.

From now on whenever we make measure-theoretic statements about $T$ on $\Omega$, such as, for example, existence of Lyapunov exponents almost surely, we implicitly assume that the statement is made with respect to the measure $\mu$. We shall refer to the dynamical system $T|_\Omega$ by $(T, \Omega)$, and $(T,\Omega,\mu)$ when the measure $\mu$ needs to be emphasized.

\subsection{Lyapunov exponents for \texorpdfstring{$\GL(k,\R)$}{} cocycles over \texorpdfstring{$(T,\Omega)$}{}}\label{subsec:lyapunov}

Define a measurable map $G: \Omega\rightarrow \GL(k, \R)$. For $n\in\N$, define $g_n: \Omega\rightarrow \R$ by
\begin{align}\label{eq:sa-process}
g_n(\omega) = \log\norm{G(T^{n-1}\omega)G(T^{n-2}\omega)\cdots G(T\omega)G(\omega)}.
\end{align}
Observe that $\set{g_n}_{n\in\N}$ defines a subadditive process in the following sense. For any $\omega\in\Omega$,
\begin{align*}
g_{n+m}(\omega) \leq g_n(\omega) + g_m(T^n\omega).
\end{align*}
According to the quite general \textit{subadditive ergodic theorem} of J. F. C. Kingman (see \cite{Kingman1973}), the limit
\begin{align*}
\lim_{n\rightarrow\infty}\frac{1}{n}g_n(\omega)
\end{align*}
exists $\mu$-almost surely, is constant $\mu$-almost surely, and belongs to $[-\infty, \infty)$. In the literature this limit is often called the \textit{Lyapunov exponent of the} $\GL(k,\R)$ \textit{cocycle over} $(T,\Omega,\mu)$. By historical convention, the Lyapunov exponent is often denoted by $\gamma$; however, we reserve the notation $\gamma$ for other things, and denote the Lyapunov exponent by $\mathcal{L}$:
\begin{align}\label{eq:lyap-exp}
\mathcal{L}:=\lim_{n\rightarrow\infty}\frac{1}{n}g_n(\omega)\hspace{2mm}\mu\text{-almost surely in }\omega.
\end{align}
In what follows, we shall be concerned with the situation where 
\begin{itemize}

\item $G: \Omega\rightarrow \GL(2,\R)$;
\item For every $\omega\in\Omega$, $G(\omega)$ has strictly positive entries;
\item $G(\omega)$ depends only on the first term of the sequence $\omega$, hence $G$ is continuous.

\end{itemize}
Then the following theorem by P. Walters (see \cite{Walters1984}) guarantees existence and finitude of the Lyapunov exponent not only almost surely, but everywhere.

\begin{thm}[P. Walters, 1984]\label{thm:walters}
Let $T: \Omega\rightarrow \Omega$ be a uniquely ergodic homeomorphism of a compact metrizable space $\Omega$ and let $G: \Omega\rightarrow \GL(k,\R)$ be continuous. If $[G(\omega)]_{ij} > 0$ for all $\omega\in\Omega$ and $1\leq i,j\leq k$, then $\mathcal{L}$ in \eqref{eq:lyap-exp} exists, is constant and finite for all $\omega\in \Omega$. Moreover, if $\mu$ denotes the unique ergodic measure for $(T,\Omega)$, then
\begin{align}\label{eq:walters}
\mathcal{L} = \lim_{n\rightarrow\infty}\frac{1}{n}\int g_n(\omega)d\mu(\omega),
\end{align}
where $\set{g_n}_{n\in\N}$ is as defined in \eqref{eq:sa-process}.
\end{thm}

Let us mention quickly that according to A. Hof's arguments \cite[Proposition 5.2]{Hof1993}, for substitution dynamical systems, such as $(T,\Omega)$, generated by primitive and invertible substitutions, such as the Fibonacci substitution---see \cite{Queffelec2010} for definitions and details---$\mathcal{L}$ exists and is constant for all $\omega\in\Omega$ even when $G(\omega)$ contains zero entries; as long as there exists $k\geq 1$ such that $G(\omega)$ depends only on the finitely many terms $\set{\omega_{-k},\dots,\omega_0,\dots\omega_k}$ of the sequence $\omega$.

\section{Description of the models and main results}\label{sec:model-and-results}

We begin with a discussion of the quantum Ising model on a one-dimensional integer lattice, with constant transverse magnetic field and interaction couplings modulated by a quasi-periodic sequence chosen freely from the hull $\Omega$ of the Fibonacci substitution sequence $u$, as constructed in the previous section.

\subsection{One-dimensional quantum Fibonacci Ising model}\label{subsec:quantum-model}

Let $p:\mathcal{A}\rightarrow \R$, with $p(a),p(b) > 0$ and, unless stated otherwise, $p(a)\neq p(b)$. For $\omega\in\Omega$, set $J_{\omega,n} = p([T^n\omega]_1)$, and define a \textit{quantum Ising Hamiltonian} on a lattice of size $N$ by
\begin{align}\label{eq:finite-quantum-ising}
\mathcal{H}_{\omega} = -\sum_{n=1}^N J_{\omega,n}\sigma_n^{(x)}\sigma_{n+1}^{(x)} - h\sum_{n=1}^N\sigma_n^{(z)},
\end{align}
where $h > 0$ is the external magnetic field in the direction transversal to the lattice. The matrices $\sigma_n^{(x),(z)}$ are spin-$1/2$ operators in the $x$ and $z$ direction, respectively, given by
\begin{align}\label{eq:tensor-def}
\sigma_n^{(x),(z)}=\underbrace{\mathbb{I}\otimes\cdots\otimes\mathbb{I}}_{N - n\text{ terms }}\otimes\underbrace{\sigma^{(x),(z)}}_{\text{$n$th position}}\otimes\underbrace{\mathbb{I}\otimes\cdots\otimes\mathbb{I}}_{N-n + 1\text{ terms }},
\end{align}
where $\mathbb{I}$ is the $2\times 2$ identity matrix. Here $\sigma^{(x),(z)}$ are the Pauli matrices given by
\begin{align*}
\sigma^{(x)} = \begin{pmatrix} 0 & 1\\ 1 & 0\end{pmatrix}\hspace{4mm}\text{and}\hspace{4mm}\sigma^{(z)}=\begin{pmatrix}1 & 0\\ 0 & -1\end{pmatrix}
\end{align*}
that we view as operators on $\C^2$. Thus $\mathcal{H}_\omega$ is a linear operator on tensor-product of $N$ copies of $\C^2$. We shall use periodic boundary conditions.

Scaling by $1/h$, by abuse of notation we can write $J_{\omega,n} = J_{\omega,n}/h$, so that the Hamiltonian in \eqref{eq:finite-quantum-ising} can be written as
\begin{align}\label{eq:finite-quantum-ising-2}
\mathcal{H}_\omega = -\sum_{n=1}^N\left( J_{\omega,n}\sigma_n^{(x)}\sigma_{n+1}^{(x)} + \sigma_{n}^{(z)}\right).
\end{align}
The Hamiltonian in \eqref{eq:finite-quantum-ising-2} can then be extended periodically to the infinite \textit{periodic} lattice with unit cell $[\omega]_{1}\cdots[\omega]_{N}$. Let us denote the resulting Hamiltonian by $\widehat{\mathcal{H}}_{\omega, N}$.

The spin model in \eqref{eq:finite-quantum-ising-2} can be transformed into a fermion model by applying the so-called Jordan-Wigner transformation (dating at least as far back as 1928 to the paper of P. Jordan and E. Wigner on second quantization \cite{Jordan1928}). This transformation transforms the Pauli spin operators into Fermi creation and annihilation operators, which in turn allows to diagonalize the Hamiltonian \eqref{eq:finite-quantum-ising-2}, and hence $\widehat{H}_{\omega, N}$. The details of this method, due to E. Lieb, T. Schultz and D. Mattis, are rather technical and can be found in \cite{Lieb1961} (for the specialization of the results from \cite{Lieb1961} to our situation, see \cite{Doria1988}). As a consequence, the energy spectrum of $\widehat{\mathcal{H}}_{\omega, N}$ can be computed as the set of those values $\lambda$ that satisfy the equations
\begin{align}\label{eq:ising-matrices}
&(A - B)\phi = \lambda\psi,\\
&(A + B)\psi = \lambda\phi,\notag
\end{align}
where $A$ and $B$ are infinite tridiagonal periodic matrices given by
\begin{align*}
A_{ii} &= -2, &A_{i, i+1} &= A_{i+1, i} = -J_{\omega,i},\\
B_{i,i+1} &= - J_{\omega,i}, &B_{i+1,i} &= J_{\omega,i}
\end{align*}
(all other entries being zero) for $1 \leq i \leq N$, and $A_{i + N, j + N} = A_{i,j}$, $B_{i + N, j + N} = B_{i,j}$, and the corresponding wave function $(\phi, \psi) = \set{(\phi_n,\psi_n)\in\C^2}_{n\in\N}$ does not diverge exponentially (for more details see \cite[Section B]{Lieb1961} and \cite{Doria1988}).

Observe that \eqref{eq:ising-matrices} is equivalent to 
\begin{align*}
(A - B)(A + B)\psi = \lambda^2\psi.
\end{align*}
This equation can be written as
\begin{align*}
\widetilde{J}_{\omega,n-1}\psi_{n-1} + (1 + \widetilde{J}^2_{\omega,n})\psi_n + \widetilde{J}_{\omega,n+1}\psi_{n+1} = \lambda^2\psi_n,
\end{align*}
where $\widetilde{J}_{\omega,i} = J_{\omega,i}$ for $1 \leq i \leq N$, and $\widetilde{J}_{\omega,i + N} = J_{\omega,i}$. Define an operator $\mathcal{J}^{(N)}_\omega$ by
\begin{align*}
[\mathcal{J}^{(N)}_\omega\psi]_n = \widetilde{J}_{\omega,n-1}\psi_{n-1} + (1 + \widetilde{J}^2_{\omega,n})\psi_n + \widetilde{J}_{\omega,n+1}\psi_{n+1}.
\end{align*}
The operator $\mathcal{J}^{(N)}_\omega$ is an infinite periodic Jacobi operator of period $N$, and as $N\rightarrow\infty$, $\mathcal{J}^{(N)}_\omega$ converges in strong sense to the operator $\mathcal{J}_\omega$ defined by
\begin{align}\label{eq:limit-hamiltonian}
[\mathcal{J}_\omega\psi]_n = J_{\omega,n-1}\psi_{\omega,n-1} + (1 + J^2_n)\psi_n + J_{\omega,n+1}\psi_{n+1}
\end{align}
(for more details see \cite{Yessen2011a}). Consequently, the spectrum of $\widehat{\mathcal{H}}_{\omega,N}$ can be investigated as the spectrum of $\mathcal{J}^{(N)}_\omega$, which, by Floquet theory, consists of $N$ compact intervals (see, for example, \cite{Toda1981}). Let us denote the spectrum of $\widehat{\mathcal{H}}_{\omega,N}$ by $\sigma_{\omega,N}$. We are interested in the structure of the spectrum in the thermodynamic limit (that is, as $N\rightarrow\infty$).

Before stating one of our main results, let us quickly recall the definition of the Hausdorff metric, as a metric on $2^\R$ and the definition of local Hausdorff dimension.

For any $A, B\subset\R$, define the Hausdorff metric $\hdist(A, B)$ by
\begin{align*}
\hdist(A, B) = \max\set{\adjustlimits\sup_{a\in A}\inf_{b\in B}\set{|a -b|}, \adjustlimits\sup_{b\in B}\inf_{a\in A}\set{|a - b|}}.
\end{align*}

Based on the connection mentioned above between $\sigma_{\omega,N}$ and the spectrum of $\mathcal{J}^{(N)}_\omega$, one may postulate a connection between the spectrum of $\mathcal{J}_\omega$ in \eqref{eq:limit-hamiltonian} and the spectrum of $\widehat{\mathcal{H}}_{\omega,N}$ in the thermodynamic limit. This indeed follows, since the Ising model is equivalent, via a unitary transformation, to the tight binding model in \eqref{eq:limit-hamiltonian} (see, for example, \cite{Satija1990, Kolar1989, Satija1994} for an attempt to approach Ising models via these tight binding models---mostly numerically). 

It was shown in \cite{Yessen2011a} that the spectrum of $\mathcal{J}_\omega$ is independent of $\omega\in\Omega$. The same independence result follows for $\widehat{\mathcal{H}}_{\omega,N}$ in the thermodynamic limit. It is therefore enough to study $\widehat{\mathcal{H}}_{\omega_s, N}$, where $\omega_s\in\Omega$ is as defined in Section \ref{subsec:fibonacci-substitution}.

Denote by $\sigma_k$ the spectrum of $\widehat{\mathcal{H}}_{\omega_s, F_k}$, where $F_k$ is the $k$th Fibonacci number. It was shown in \cite{Yessen2011a}, that as $k\rightarrow\infty$, the spectra of $\mathcal{J}_{\omega_s}^{(F_k)}$ approach in Hausdorff metric the spectrum of $\mathcal{J}_{\omega_s}$. Again, via the aforementioned connection between the Ising Hamiltonian and Jacobi matrices, one would expect a similar result for the Ising Hamiltonian. This is indeed the case.

To avoid a rather opaque technical thicket, we shall proceed directly by explicitly constructing a compact set and showing that the sequence $\set{\sigma_k}$ approaches this set in the Hausdorff metric. We then proceed to describe topological, measure-theoretic and fractal-dimensional properties of this limit set, which we call the \textit{energy spectrum of} $\mathcal{H}_{\omega_s, N}$ in the thermodynamic limit.

Before we continue, let us introduce the following notation. We shall denote the Hausdorff dimension of a set $A \subset \R$ by $\hdim(A)$. By $\lhdim(A, a)$ we denote the local Hausdorff dimension of $A$ at a point $a$: 
\begin{align*}
\lhdim(A, a) = \lim_{\epsilon\rightarrow 0}\hdim(A\cap(a-\epsilon, a + \epsilon)).
\end{align*}
The box-counting dimension of $A$ is denoted by $\bdim(A)$.

\begin{rem}
In what follows, by a \textit{Cantor set} we mean a (nonempty) compact set with no isolated points whose complement is dense.
\end{rem}

Since a certain multifractal with specific properties will appear more than once, it is convenient to give this set a name:

\begin{defn}\label{defn:um-set}
We call a Cantor set $S\subset \R$ a \textit{um-set} (for \textit{uniform multifractal set}) provided that $S$ satisfies the following properties.

\begin{enumerate}[(a)]

\item The set $S$ has zero Lebesgue measure;

\item The function $d$ that maps $s\in S$ to $\lhdim(S, s)$ is continuous and for all $\epsilon > 0$ and for $s\in S$, $d$ restricted to $(s - \epsilon, s + \epsilon)\cap S$ is non-constant (in particular, $S$ is "uniformly"---i.e. localized at every $s\in S$---a multifractal);

\item The Hausdorff dimension of $S$ is strictly between zero and one.

\end{enumerate}
 
If (c) fails, we call $S$ \textit{full um-set}. If $S$ satisfies (a)--(c) and, in addition,

\begin{enumerate}[(d)]

\item The box counting dimension of $S$ exists and coincides with its Hausdorff dimension,

\end{enumerate}

then we call $S$ a \textit{moderate um-set}. If $S$ satisfies (a)--(d) except (c), we call $S$ a \textit{full moderate um-set}.

\end{defn}
\begin{rem}
 The terminology \textit{moderate set} is not standard and, to the best of the author's knowledge, has not appeared (at least in this context) before. This terminology is introduced here only for convenience.
\end{rem}

Let us recall that, roughly speaking, a \textit{dynamically defined} Cantor set $C$ is one which is a blowup of an arbitrarily small neighborhood of any $a\in C$ under a $C^{1 + \epsilon}$ map defined in a neighborhood of $C$ (here $C^{1 + \epsilon}$ means continuously differentiable with H\"older continuous derivative whose H\"older exponent is $\epsilon$). For the technical definition and properties of dynamically defined Cantor sets, see \cite[Chapter 4]{Palis1993}. Let us only mention that dynamically defined Cantor sets are rather rigid. For example, the box-counting dimension exists and coincides with the Hausdorff dimension. Moreover, local box-counting and Hausdorff dimensions do not depend on the point of localization and coincide with the global dimensions. 

It is convenient to give a name to a Cantor set which is, in a sense, arbitrarily close to being dynamically defined:
\begin{defn}\label{defn:almost-dynamically-defined}
We call a Cantor set $S$ \textit{almost dynamically defined Cantor set} provided that there exists a finite subset $\mathcal{T}$ of $S$ (possibly empty) such that outside an arbitrarily small neighborhood of $\mathcal{T}$, the remainder of $S$ is a dynamically defined Cantor set.
\end{defn}
We are now ready to state our main results for the quantum Ising quasicrystal.
\begin{thm}\label{thm:spectral-convergence-thm}
There exists a compact (nonempty) set $B_\infty$ in $\R$ such that the sequence $\set{\sigma_k}_{k\in\N}$ converges to $B_\infty$ in the Hausdorff metric. The set $B_\infty$ implicitly depends on the choice of $p(a)$ and $p(b)$, and has the following properties.
\begin{enumerate}[i.]

\item $B_\infty$ is a um-set, but not a full um-set;

\item With $p(a)$ fixed, for all choices of $p(b)$ sufficiently close to $p(a)$, but not equal to $p(a)$, $B_\infty$ is a moderate um-set, but not a full moderate um-set;

\item $\hdim(B_\infty)$ depends continuously on $(p(a), p(b))$.

\end{enumerate}
\end{thm}
\begin{rem}
Thus Theorem \ref{thm:spectral-convergence-thm} is an extension of \cite[Theorem 2.1]{Yessen2011}. Also, we believe the restriction on the parameter $p(b)$ in statement ii can be dropped (see the concluding remarks in \cite{Yessen2011a} for more details).
\end{rem}

Next we discuss the topological structure of the energy spectrum for many noninteracting fermions. This can be thought of as $N$ identical Ising models on separate $1$D lattices with no interaction between them. For $N > 1$  noninteracting fermions, the corresponding energy spectrum is the $N$--fold arithmetic sum of single-fermion spectra, excluding double-counting (as two fermions cannot occupy the same state). More precisely, if we denote by $B_\infty^N$ the $N$--fermion spectrum, then we have
\begin{align*}
B_\infty^N = \overline{\set{a_1 + \cdots + a_N: a_i\in B_\infty\text{ and for any } i\neq j, a_i\neq a_j}},
\end{align*}
where $\overline{\bullet}$ denotes topological closure. On the other hand, since, as claimed in Theorem \ref{thm:spectral-convergence-thm}, $B_\infty$ is a Cantor set, it contains no isolated points. Therefore,
\begin{align*}
B_\infty^N = \set{a_1 + \cdots + a_N: a_i\in B_\infty}.
\end{align*}
\begin{rem}
 In this context, the term \textit{energy spectrum} is synonymous with the more common \textit{excitation spectrum}. 
\end{rem}

Regarding topological structure of $B_\infty^N$, we have the following.

\begin{thm}\label{thm:thickness}
For every choice of positive and distinct $p(a)$, $p(b)$, there exists an $N_0\in\N$ such that for all $N\geq N_0$, $B_\infty^N$ is the interval $[A, B]$, where
\begin{align*}
A = \min\set{\sum_{i = 1}^N a_i, a_i \in B_\infty}\hspace{4mm} B = \max\set{\sum_{i = 1}^N b_i, b_i \in B_\infty}.
\end{align*}
Consequently, for $N\geq N_0$, 
\begin{align}\label{eq:thm-thickness}
B_\infty^N = [-Nb, Nb], 
\end{align}
where $b$ is the upper bound of $B_\infty$; hence if $N = \infty$, then $B_\infty^N = \R$. Moreover, for any $p(a) > 0$ and all $p(b) > 0$ sufficiently close to $p(a)$ and not equal to $p(a)$, we can take $N_0 = 2$.
\end{thm}
\begin{rem}
The topological structure of $B_\infty^N$ for $N < N_0$ is an interesting question. At present we do not have any strong qualitative results in this direction. We believe that when $p(a)$ and $p(b)$ are sufficiently far, then for some $N \geq 2$, $B_\infty^N$ is a Cantorval---a compact set with dense interior, any connected component of the complement of which has its endpoints approximated by the endpoints of other connected components of the complement. In other words, it is in a sense a "Cantor set with dense interior." Sets of this nature, and their appearence in dynamical systems and number theory, have been under active investigation.
\end{rem}

\subsection{One-dimensional classical Fibonacci Ising model}\label{subsec:subsec-classical}

We first consider the classical 1D Ising spin chain in a real magnetic field and real temperature regime. For the sake of brevity, we refer to it simply as the \textit{real regime}, and in the complex magnetic field case as the \textit{complex regime}.

\subsubsection{Real regime}\label{subsubsec:real-regime}

Let us consider a linear chain of $N$ spins $\sigma_i\in\set{\pm 1}$, $1\leq i \leq N$. The energy of a configuration $\sigma = (\sigma_1,\sigma_2,\dots,\sigma_N)$ is given by
\begin{align}\label{eq:eq-energy-real}
\mathcal{H}(\sigma) = -\sum_{i = 1}^N\left(J_i\sigma_i\sigma_{i+1} + H_i\sigma_i\right),
\end{align}
with periodic boundary conditions (i.e. $\sigma_{N+1} = \sigma_1$). As mentioned in the introduction, we concentrate our attention on the ferromagnetic case (that is, the interaction strength couplings $J_i$ are positive). We take $H_i$, the external magnetic field, non-negative. 

The partition function of this system is given by
\begin{align}\label{eq:eq-real-partition}
\mathbf{Z}^{(N)} = \sum_{\sigma}\exp\left(-\frac{1}{k_B\tau}\mathcal{H}(\sigma)\right),
\end{align}
where the sum is taken over all possible configurations of $N$ spins. Here $\tau$ denotes the temperature and $k_B$ is the Boltzmann's constant. The corresponding free energy is given by
\begin{align}\label{eq:eq-real-free-energy}
\mathbf{F}^{(N)} = -\frac{1}{Nk_B\tau}\log\mathbf{Z}^{(N)},
\end{align}
expressed as a function of $\tau$. 

Let $p: \mathcal{A}\rightarrow\R$, $p(a), p(b) \geq 0$ and $p(a)\neq p(b)$ unless stated otherwise. Also, introduce $q: \mathcal{A}\rightarrow\R$ with $q(a),q(b)\geq 0$ (we do not force, in general, that $q(a)\neq q(b)$). For any $\omega\in\Omega$, $n\in\N$, set
\begin{align*}
J_{\omega,n} = p([T^n\omega]_1)\hspace{2mm}\text{ and }\hspace{2mm} H_{\omega,n} = q([T^n\omega]_1).
\end{align*}
With these sequences of interaction strength couplings, $\set{J_{\omega,n}}$, and magnetic field, $\set{H_{\omega,n}}$, we obtain a classical (i.e. non-quantum) quasi-periodic Ising model in \eqref{eq:eq-energy-real}, which depends on the choice of $\omega\in\Omega$ (\textit{quasi-periodic} in the sense that the interaction couplings and the magnetic field are modulated by a quasi-periodic sequence $\omega$, and \textit{periodic} in the sense that we use periodic boundary conditions). To emphasize this dependence, let us write $\mathbf{Z}^{(N)}_\omega$ and $\mathbf{F}^{(N)}_\omega$ for the corresponding partition function and the free energy, respectively, on a lattice of size $N$. 

Let us remark that the partition function $\mathbf{Z}^{(N)}$ is obviously strictly positive, so that $\mathbf{F}^{(N)}$ is well-defined. Regularity properties of $\mathbf{F}^{(N)}$ in the thermodynamic limit (that is, as $N\rightarrow\infty$), which we denote by $\mathbf{F}$ (or by $\mathbf{F}_\omega$ when dependence on $\omega\in\Omega$ needs to be emphasized), presents a curious problem in connection with critical phenomena (i.e. phase transitions). Since $\mathbf{Z}^{(N)}$ does not admit zeros for finite $N$ and $0 < \tau < \infty$, the free energy does not experience critical behavior in finite volume and finite temperature (this is of course expected in dimension one). In fact, for finite $N$ and $0 < \tau < \infty$, $\mathbf{F}^{(N)}$ is easily seen to be analytic. Does this property survive the thermodynamic limit? It is expected that the answer is yes. In some cases $\mathbf{F}$ can be computed analytically; in the present case, however, properties of $\mathbf{F}_\omega$ (assuming $\mathbf{F}_\omega$ 
is well-defined) were hitherto not known (see, for example, remarks in \cite{Baake1999}). This brings us to

\begin{thm}\label{thm:thm-classical-analyticity}
The sequence of functions $\mathbf{F}_\omega^{(N)}$ converges, as $N\rightarrow\infty$, uniformly in $\tau$ to an analytic function $\mathbf{F}$ of $\tau\in(0,\infty)$ which is independent of $\omega$ and is strictly negative.
\end{thm}

\begin{rem}
Our proof of Theorem \ref{thm:thm-classical-analyticity} also provides a recursive algorithm for the numerical approximation of $\mathbf{F}$.
\end{rem}

\subsubsection{Complex regime}\label{subsubsec:complex-regime}

In this case we take $q:\mathcal{A}\rightarrow \C$ constant---denote this constant by $H$---and $0 < \tau < \infty$. For a choice of $\omega\in\Omega$, the partition function $\mathbf{Z}^{(N)}_\omega$ can then be considered as a function of three variables:
\begin{align}\label{eq:complex-regime-variables}
\alpha = \exp\left(\frac{p(a)}{k_B\tau}\right),\hspace{4mm} \beta = \exp\left(\frac{p(b)}{k_B\tau}\right),\hspace{4mm}\eta = \exp\left(\frac{H}{k_B\tau}\right)
\end{align}
(that is, $\alpha, \beta$ and $\eta$ are functions of $\tau$: $\alpha,\beta: \tau\mapsto \R$, $\eta: \tau\mapsto\C$). From now on we shall concentrate on fixed $\tau\in (0,\infty)$ and variable $H\in\C$ (that is, $\alpha,\beta$ are fixed, while the \textit{fugacity variable} $\eta$ is varied over $\C$). 

While $\mathbf{Z}^{(N)}_\omega$ does not admit zeros in the real regime, it does (even in the finite volume $N$ and finite $\tau\in(0,\infty)$) in the complex regime, as a function of $\eta$ (indeed, $\mathbf{Z}^{(N)}_\omega(\eta)$ is a polynomial in $\eta$). These so-called Lee-Yang zeros \cite{Lee1952} are guaranteed to lie on the unit circle $S^1\subset\C$ if both couplings are ferromagnetic (that is, $p(a),p(b) \geq 0$, which is the case considered here); similarly, in the purely antiferromagnetic case (that is, $p(a), p(b) < 0$), the zeros lie on the negative real axis; the mixed case is much more subtle and is not treated here (for further details and a more general exposition see, for example, \cite{Ruelle1989}). We are interested here in distribution of zeros of $\mathbf{Z}^{(N)}_\omega$ in the thermodynamic limit, $N\rightarrow\infty$.

Let $\mathcal{Z}_{\omega,N}$ denote the set of zeros of $\mathbf{Z}^{(N)}_\omega$ in the variable $\eta$. Some numerical experiments and soft analysis in \cite{Baake1995} suggest that in the limit $N\rightarrow\infty$, $\set{\mathcal{Z}_{\omega_s,N}}$ accumulates in a Cantor set on $S^1$, where $\omega_s\in\Omega$ is as defined in Section \ref{subsec:fibonacci-substitution}. Also, a gap structure similar to that studied by J. Bellissard et. al. in \cite{Bellissard1992} was observed. We rigorously justify these observations, and add further detail to the qualitative description of the zeros in the thermodynamic limit in the following 

\begin{thm}\label{thm:complex-ising-zeros}
There exists a compact, nonempty set $\mathcal{Z}\subset S^1\subset\C$ such that $\mathcal{Z}_{\omega_s,N}\xrightarrow[N\rightarrow\infty]{}\mathcal{Z}$ in the Hausdorff metric. Moreover, the set $\mathcal{Z}$ has the following properties.
\begin{enumerate}[i.]

\item $\mathcal{Z}$ is a um-set together with (possibly) finitely many isolated points;

\item $\set{\pm1}\cap\mathcal{Z} = \emptyset$ (as expected, of course);

\item With $p(a)$ fixed, for all $p(b)$ sufficiently close to $p(a)$, $\mathcal{Z}$ is a moderate um-set;

\end{enumerate}
\end{thm}
\begin{rem}
We believe that in statement i, there do not exist the mentioned isolated points, though at present we are unable to prove it; however, iii states that no isolated points are present provided that $p(a)/p(b)\approx 1$ (assuming $p(b)\neq 0$).
\end{rem}

Since in Theorem \ref{thm:complex-ising-zeros} statements are made only for the special choice $\omega_s\in\Omega$, a natural question still remains: how is the conclusion of Theorem \ref{thm:complex-ising-zeros} affected as one considers general $\omega\in\Omega$? At present we do not have a general answer, but we conjecture that the conclusion of the theorem is not affected by any specific choice of $\omega\in\Omega$, and, moreover, $\mathcal{Z}$ in Theorem \ref{thm:complex-ising-zeros} is independent of the choice of $\omega\in\Omega$.

\section{Proofs of main results}\label{sec:proofs}

For proofs of main results we rely heavily on dynamics of the so-called \textit{Fibonacci trace map} (to be defined momentarily), and the geometry of some dynamical invariants. We recall here only those properties of the aforementioned objects that are necessary to prove our results, and only briefly. For a survey on trace maps (including the Fibonacci trace map) associated to substitution sequences, the reader may consult, for example, \cite{Roberts1994, Roberts1994b, Baake1997}. For the necessary notions from hyperbolic and partially hyperbolic dynamics (including notation and terminology), see a brief outline in Appendix \ref{b}, and references therein for more details. For model-independent results (some of which appear below), see \cite[Section 2]{Damanik0000my1}.

\subsection{Dynamics of the Fibonacci trace map: a brief outline}\label{subsec:fibonacci-dynamics}

Define the so-called \textit{Fibonacci trace map} on $\C^3$ by
\begin{align}\label{eq:trace-map}
f(x,y,z) = (2xy - z, x, y).
\end{align}
Define also the so-called \textit{Fricke-Vogt character} (often also called \textit{Fricke-Vogt invariant}) \cite{Fricke1896, Fricke1897, Vogt1889} $I: \C^3\rightarrow\C$ by
\begin{align}\label{eq:fv-invariant}
I(x,y,z) = x^2 + y^2 + z^2 - 2xyz - 1.
\end{align}
It turns out that $I$ is invariant under the action of $f$, that is, 
\begin{align*}
I(x,y,z) = I\circ f(x,y,z)
\end{align*}
(in connection with Schr\"odinger operators, see the pioneering work in \cite{Kadanoff1984, Kadanoff0000, Kohmoto1983, Kohmoto1987b, Ostlund1983}). Since $f$ is invertible, it follows that the sets
\begin{align*}
S_{V\in\C} = \set{(x,y,z)\in\C^3: I(x,y,z) = V}
\end{align*}
are also invariant under $f$, and in fact $f(S_V) = S_V$. The sets $S_V$ are complex-analytic surfaces.

In what follows, we sometimes consider the real part of $S_V$ separately (that is, the subset $S_V\cap \R^3$ of $S_V$). Moreover, as far as the real part, we only need to consider $S_V\cap \R^3$ for $V\in \R$ and $V \geq 0$. In this case, if $V > 0$, then $S_V\cap \R^3$ is a smooth, connected, unbounded two-dimensional submanifold of $\R^3$ without boundary. It is homeomorphic to the two-dimensional sphere with four points removed. If $V = 0$, then $S_V\cap \R^3$ has four conic singularities. Minus these singularities, it is a smooth manifold with five connected components: one bounded, diffeomorphic to the two-dimensional sphere with four points removed, and four unbounded, diffeomorphic to the two-dimensional disc with one point removed. 

To distinguish between the "full complex surface" $S_V$ and $S_V\cap \R^3$, we denote the former by $S_V^\C$ and the latter by $S_V^\R$. The surface $S_0^\R$ is called the \textit{Cayley cubic}.

It isn't difficult to see that $S_V^\R$ is also invariant under $f$. In what follows, we are concerned mostly with those points in $S_V^{\C}$, that have bounded forward semi-orbit under the action of $f$; that is, $p\in S_V^{\C}$ such that
\begin{align*}
\mathcal{O}^+_f(p):= \set{f^n(p): n\in\N}
\end{align*}
is bounded (we define the backward semi-orbit of $p$, $\mathcal{O}^-_f(p)$, similarly by considering iterates of $f^{-1}$ in place of iterates of $f$). Following convention from the theory of holomorphic dynamical systems, let us define the following sets:
\begin{align*}
K^+_V := \set{p\in S_V^\C: \mathcal{O}^+_f\text{ is bounded}}\hspace{2mm}\text{ and }\hspace{2mm}K^-_V := \set{p\in S_V^\C: \mathcal{O}^-_f\text{ is bounded}}.
\end{align*}
Define also
\begin{align*}
\Omega_V := K^+_V\cap K^-_V.
\end{align*}

It turns out that $\Omega_V\subset S_V^\R$ (see \cite{Cantat2009}). Also, if a point has unbounded forward semi-orbit, then the orbit escapes to infinity, i.e. the orbit does not contain a bounded subsequence (see, for example, \cite{Roberts1996}).

Regarding the topological structure of $\Omega_V$ and dynamics of $f$ restricted to $\Omega_V$ (clearly $\Omega_V$ is $f$-invariant) we have the following theorem. For definition of hyperbolicity, see Section \ref{b1}. 

\begin{thm}[M. Casdagli \cite{Casdagli1986}, D. Damanik and A. Gorodetski \cite{Damanik2009}, S. Cantat \cite{Cantat2009}]\footnote{M. Casdagli proved this theorem for $V > 0$ sufficiently large, D. Damanik \& A. Gorodetski extended Casdagli's result to all $V > 0$ sufficiently small, and S. Cantat proved the result for all $V > 0$. M. Casdagli and D. Damanik \& A. Gorodetski used Alekseev's cone technique to prove hyperbolicity, while S. Cantat used other techniques from holomorphic dynamics.}\label{thm:cas-gor-cant}
For all $V > 0$, $\Omega_V\subset S^\R_V$ is a compact Cantor set of zero Lebesgue measure. It is precisely the nonwandering set for $f$ restricted to $S^\R_V$ and $f$ on $\Omega_V$ is hyperbolic (the nonwandering set is the set of those points $x$ such that for every open neighborhood $U$ of $x$ and $N\in\N$ there exists $n > N$ such that $f^n(U)\cap U\neq \emptyset$). Moreover, $\Omega_V$ is locally maximal and $f|_{\Omega_V}$ is topologically transitive.
\end{thm}

From \cite{Yessen2011} we have the following (see Section \ref{b2} for definition of partial hyperbolicity).

\begin{prop}[{\cite[Combination of Corollary 4.8 and Proposition 4.9]{Yessen2011}}]\label{prop:phyperb}
For any $J = [V_1, V_2]$, with $V_1 > 0$, $\bigcup_{V\in J}\Omega_V$ is partially hyperbolic for $f$. Thus, $\bigcup_{V\in(V_1, V_2)}K_V^+$ is a family of pairwise disjoint smooth open two-dimensional injectively immersed submanifolds of $\R^3$ without boundary, called the \textit{center-stable manifolds}. The center-stable manifolds intersect the surfaces $S_V^\R$, $V > 0$, transversally.
\end{prop}

%

We shall also need regularity of escape rate of those points whose forward semi-orbit is unbounded (the set of such points is open in $\C^3$ \cite{Cantat2009}).

\begin{prop}\label{prop:escape-rate}
Let $D$ be a domain in $\C^3$ such that no point of $D$ has bounded forward semi-orbit under $f$. Then the map $\mathcal{E}: D\rightarrow \R$ defined by
\begin{align}\label{eq:anal-limit}
\mathcal{E}: z \mapsto \lim_{n\rightarrow\infty}\frac{\log\norm{f^n(z)}}{\phi^n},
\end{align}
where $\phi = (\sqrt{5}+1)/2$ is the golden mean, is well-defined and analytic. Moreover, $\mathcal{E}(z) = \widetilde{\mathcal{E}}(z)$, where
\begin{align}\label{eq:anal-limit-2}
\widetilde{\mathcal{E}}: z\mapsto\lim_{n\rightarrow\infty}\frac{\log\abs{\pi\circ f^n(z)}}{\phi^n},
\end{align}
where $\pi$ denotes projection onto the first coordinate. We assume that $\pi\circ f^n(z)\neq 0$, which is guaranteed for sufficiently large $n$ by hypothesis.
\end{prop}

\begin{proof}[Proof of Proposition \ref{prop:escape-rate}]

Existence of the limit in \eqref{eq:anal-limit} follows from \cite[Corollary 3.3]{Cantat2009}. That the sequence $\log\norm{f^n(z)}/\phi^n$ is uniformly Cauchy in $z$, and hence converges uniformly in $z$, follows from arguments similar to those given in the proof of \cite[Lemma 3.2]{Bedford1991}. Since for every $n\geq 1$, the map $z\mapsto \log\norm{f^n(z)}/\phi^n$ is pluriharmonic on $D$, it follows that its uniform limit, $\mathcal{E}$, is also pluriharmonic, and hence analytic, on $D$. 

The second assertion (that is, \eqref{eq:anal-limit-2}) follows from the fact that the escape rate of orbits under $f$ is the same as the escape rate in each coordinate (for details, see, for example, \cite{Cantat2009, Roberts1996}). 
\end{proof}

Next we state one of the main results of \cite{Yessen2011a} in a slightly modified, more general geometric form. This result will play a central role in the proof of Theorem \ref{thm:complex-ising-zeros} below. Although proofs in \cite{Yessen2011a} can be easily extended to cover this more general case, we invite the reader to see a discussion on model-independent results in \cite[Section 2]{Damanik0000my1} for the technical details.

\begin{thm}[{\cite[Theorem 2.3]{Yessen2011a}}]\label{thm:thm-yessen}
Suppose that $J\subset\R$ is a compact interval and $\gamma: J\rightarrow \bigcup_{V\geq 0}S_V^\R$ is a regular analytic curve that intersects the center-stable manifolds. Then either $\gamma$ lies entirely on a center-stable manifold, or it intersects the center-stable manifolds transversally in all but, possibly, finitely many points. Denote the set of tangential intersections by $\mathcal{T}$. Let $\Gamma$ denote the set $\gamma\cap\set{\text{center-stable manifolds}}$. If  $\Gamma$ contains points outside of $\mathcal{T}$ other than just the endpoints of $\gamma$, then the following properties hold.

\begin{enumerate}[i.]

\item $\Gamma$ is a um-set together with (possibly) finitely many isolated points if and only if $\gamma$ does not lie entirely in some $S_V^\R$, for some $V > 0$; otherwise $\Gamma$ is an almost dynamically defined Cantor set;

\item If $\mathcal{T} = \emptyset$, then $\Gamma$ is a moderate um-set if $\Gamma$ does not lie entirely on some $S_V^\R$, $V > 0$, and a dynamically defined Cantor set otherwise;

\item $\Gamma\setminus\mathcal{T}$ is a full um-set if and only if there exists a point in $\Gamma\setminus\mathcal{T}$ lying on $S_0$. 

\item Suppose $\gamma$ depends analytically on a parameter $\kappa$. Write $\Gamma_\kappa$ for the corresponding set of intersections with the center-stable manifolds. If for some $\kappa_0$, $\Gamma_{\kappa_0}$ is a um-set together with (possibly) finitely many isolated points, then for all $\kappa$ sufficiently close to $\kappa_0$, $\Gamma_\kappa$ is also a um-set together with (possibly) finitely many isolated points. Moreover, if $\Gamma_{\kappa_0}$ does not contain isolated points, then for all $\kappa$ sufficiently close to $\kappa_0$, $\Gamma_{\kappa_0}$ does not contain isolated points and $\kappa\mapsto\hdim(\Gamma_\kappa)$ is continuous. However, the properties of being a full um-set and almost dynamically defined Cantor set do not survive arbitrarily small $\kappa$-perturbations. 
\end{enumerate}

\end{thm}

We are now equipped with all the tools needed to prove the main results.

\subsection{Proof of Theorem \ref{thm:spectral-convergence-thm}}\label{subsec:proof-thm-spectral-convergence}

For any choice $p(a), p(b) > 0$ (as usual, $p(a)\neq p(b)$), define the line $\gamma$ in $\R^3$ parameterized by $E\in\R$ by
\begin{align*}
\gamma(E) = \left(\frac{E^2 - (1 + p(b)^2)}{2p(b)}, \frac{E^2 - (1 + p(a)^2)}{2p(a)}, \frac{p(a)^2 + p(b)^2}{2p(a)p(b)}\right).
\end{align*}
Define also
\begin{align*}
B_\infty = \set{E\in\R: \mathcal{O}^+_f(\gamma(E))\text{ is bounded}}.
\end{align*}
It was shown in \cite{Yessen2011} that $B_\infty$ is compact (and, of course, nonempty) and that $\sigma_k\xrightarrow[k\rightarrow\infty]{}B_\infty$ in Hausdorff metric. Thus the first claim of Theorem \ref{subsec:proof-thm-spectral-convergence} follows. We now investigate properties of $B_\infty$. We start by noting that $B_\infty$ contains infinitely many points none of which are isolated (this can be seen from relation of $B_\infty$ with the spectrum of the tight-binding Hamiltonian in \eqref{eq:limit-hamiltonian}, which was investigated rigorously in \cite{Yessen2011a}).

\begin{rem}
Let us remark that properties of $B_\infty$ were investigated in \cite{Yessen2011} in the special case that for fixed $p(a) > 0$, $p(b)/p(a)$ is sufficiently close to one.
\end{rem}

It is of course enough to consider
\begin{align*}
\widetilde{\gamma}(\widetilde{E}) = \left(\frac{\widetilde{E} - (1 + p(b)^2)}{2p(b)}, \frac{\widetilde{E} - (1 + p(a)^2)}{2p(a)}, \frac{p(a)^2 + p(b)^2}{2p(a)p(b)}\right)
\end{align*}
in place of $\gamma$, with $\widetilde{E}\in[0,\infty)$. Evaluating the Fricke-Vogt invariant from \eqref{eq:fv-invariant} along $\widetilde{\gamma}$, we get
\begin{align}\label{eq:fv-at-gamma}
I(\widetilde{\gamma}(\widetilde{E})) = \frac{\widetilde{E}}{4}\left(\frac{1}{r} - r\right)^2,
\end{align}
where $r = p(a)/p(b)$. Clearly $I(\widetilde{\gamma}(\widetilde{E}))\in\R$ and $I(\widetilde{\gamma}(\widetilde{E}))\geq 0$. Moreover, we have
\begin{align*}
\frac{\partial I(\widetilde{\gamma}(\widetilde{E}))}{\partial \widetilde{E}} = \frac{1}{4}\left(\frac{1}{r} - r\right)^2 > 0.
\end{align*}
Since the image of $\gamma$ coincides with that of $\widetilde{\gamma}$, it follows that $\gamma$ is transversal to every $S_V^\R$, $V \geq 0$, and from \eqref{eq:fv-at-gamma} it follows that $\gamma$ intersects each $S_V^\R$, $V\geq 0$ in only one point. 

The points of intersection of $\widetilde{\gamma}$ with the center-stable manifolds lie on a compact line segment along $\widetilde{\gamma}$, since for all $\widetilde{E}$ such that $\abs{\widetilde{E}}$ is sufficiently large, $\widetilde{\gamma}(\widetilde{E})$ satisfies the hypothesis of \cite[Proposition 4.1-(2)]{Yessen2011}, and hence escapes. We are therefore working in the setting of Theorem \ref{thm:thm-yessen}.

Observe that $\widetilde{\gamma}$ intersects $S_0^\R$ at $\widetilde{E} = 0$, in the point
\begin{align}\label{eq:gamma-cc-intersection}
p_0 := \widetilde{\gamma}(0) = \left(-\frac{1 + p(b)^2}{2p(b)}, -\frac{1 + p(a)^2}{2p(a)}, \frac{p(a)^2 + p(b)^2}{2p(a)p(b)}\right).
\end{align}
With $(x_0, y_0, z_0) = f^2(p_0)$, a simple (albeit long) computation shows $\abs{x_0},\abs{y_0}$ and $\abs{z_0}$ are greater than one, and $\abs{x_0y_0}>\abs{z_0}$. It follows that $\mathcal{O}^+_f(f^2(p_0))$ is unbounded (see \cite[Proposition 4.1-(2)]{Yessen2011}). Thus Theorem \ref{thm:spectral-convergence-thm}-i follows from Theorem \ref{thm:thm-yessen}-iii.

In \cite{Yessen2011} it was shown that with $p(a) > 0$ fixed, for all $p(b) > 0$ sufficiently close to $p(a)$, $\gamma$ intersects the center-stable manifolds transversally. Hence Theorem \ref{thm:spectral-convergence-thm}-ii follows from Theorem \ref{thm:thm-yessen}-ii.

Again, Theorem \ref{thm:spectral-convergence-thm}-iii follows from Theorem \ref{thm:thm-yessen}, since $\gamma$ depends jointly analytically on $p(a)$ and $p(b)$.

This completes the proof of Theorem \ref{thm:spectral-convergence-thm}.

\subsection{Proof of Theorem \ref{thm:thickness}}\label{subsec:proof-thm-thickness}

First we need to introduce the notion of the \textit{Newhouse thickness} (see \cite{Newhouse1968}).

Let $\mathcal{C}\subset\R$ be a Cantor set and denote by $\mathcal{I}$ its convex hull. Any connected component of $\mathcal{I}\setminus\mathcal{C}$ is called a \textit{gap} of $\mathcal{C}$. A \textit{presentation} of $\mathcal{C}$ is given by an ordering $\mathcal{U} = \set{U_n}_{n\geq 1}$ of the gaps of $\mathcal{C}$. If $u\in\mathcal{C}$ is a boundary point of a gap of $U$ of $C$, we denote by $K$ the connected component of $\mathcal{I}\setminus(U_1\cup U_2\cup\cdots\cup U_n)$ (with $n$ chosen so that $U_n = U$) that contains $u$, and write

\begin{align*}
\tau(\mathcal{C},\mathcal{U},u) = \frac{\abs{K}}{\abs{U}}.
\end{align*}

With this notation, the \textit{thickness} $\tau(\mathcal{C})$ of $\mathcal{C}$ is given by

\begin{align*}
\tau(\mathcal{C}) = \sup_{\mathcal{U}}\inf_u\tau(\mathcal{C},\mathcal{U},u).
\end{align*}

Next we define what S. Astels in \cite{Astels2000} calls \textit{normalized thickness}. For a Cantor set $\mathcal{C}$, define the normalized thickness of $\mathcal{C}$ by
\begin{align*}
\nu(\mathcal{C}) = \frac{\tau(\mathcal{C})}{1 + \tau(\mathcal{C})}
\end{align*}
(Astels originally used the letter $\gamma$ for normalized thickness, which we use here for the line of initial conditions to preserve notation from \cite{Yessen2011, Yessen2011a} for easy reference).

A specialized version of the more general \cite[Theorem 2.4]{Astels2000} gives the following remarkable result.

\begin{thm}\label{thm:thm-astels}
Given a Cantor set $\mathcal{C}$, if $m\nu(\mathcal{C})\geq 1$, $m\in\N$, then the $m$-fold arithmetic sum of $\mathcal{C}$ with itself is equal to the $m$-fold sum of $\mathcal{I}$ with itself, where $\mathcal{I}$ is the convex hull of $\mathcal{C}$. 
\end{thm}

In order to apply Theorem \ref{thm:thm-astels} in our case, it is enough to show that for any positive and distinct $p(a)$, $p(b)$, we have $\tau(B_\infty) > 0$. 

For a fixed $a\in B_\infty$, define the \textit{local thickness} at $a$ by
\begin{align*}
\lim_{\epsilon\rightarrow 0}\tau(B_\infty\cap (a - \epsilon, a + \epsilon))
\end{align*}
(that the limit exists follows easily from the fact that the thickness is decreasing with respect to set inclusion). From \cite[Lemma 3.6]{Yessen2011a} it follows that for any $a\in B_\infty$, the local thickness at $a$ coincides with the thickness of some dynamically defined Cantor set. On the other hand, the thickness of any dynamically defined Cantor set is strictly positive (see \cite[Chapter 4]{Palis1993}). Hence $\tau(B_\infty) > 0$. This proves the first statement of Theorem \ref{thm:thickness}.

To prove the second statement, it is of course enough to show that for fixed $p(a) > 0$ and all $p(b) > 0$ sufficiently close to $p(a)$, $\tau(B_\infty)$ is bigger than one. By \cite[Theorem 2.1-iv]{Yessen2011}, we know that $\hdim(B_\infty)$ can be forced arbitrarily close to one by taking $p(b)$ sufficiently close to $p(a)$. This, together with the following relation between Hausdorff dimension and thickness (see \cite{Palis1993} for details) gives $\tau(B_\infty) > 1$:
\begin{align*}
\frac{\log 2}{\log\left(2 + \frac{1}{\tau(\bullet)}\right)}\leq \hdim(\bullet).
\end{align*}

Finally, \eqref{eq:thm-thickness} follows since $B_\infty$ is symmetric with respect to the origin in $\R$ (see proof of Theorem \ref{thm:spectral-convergence-thm} above, and notice the appearance of $E^2$ in the definition of $\gamma$).

\subsection{Proof of Theorem \ref{thm:thm-classical-analyticity}}\label{subsec:proof-thm-classical-analyticity}

Let $K_{\omega,n} = J_{\omega,n}/k_B\tau$ and $h_{\omega,n} = H_{\omega,n}/k_B\tau$. Then the partition function $\mathbf{Z}^{(N)}_\omega$  can be written as 
\begin{align}\label{eq:partition-trace-form}
\mathbf{Z}^{(N)}_\omega = \Tr\left(\mathcal{T}^{(N)}_\omega\right),\hspace{2mm}\text{ where }\hspace{2mm}\mathcal{T}^{(N)}_\omega = \prod_{i = N}^1\mathcal{T}_{\omega,i},
\end{align}
with
\begin{align}\label{eq:mat-form}
\mathcal{T}_{\omega,i} =
\begin{pmatrix}[1.5]
\exp\left(K_{\omega,i} + h_{\omega, i}\right) & \exp\left(-K_{\omega,i} - h_{\omega,i}\right)\\ \exp\left(-K_{\omega,i} + h_{\omega,i}\right) & \exp\left(K_{\omega,i} - h_{\omega,i}\right)\end{pmatrix}
\end{align}
(recall that we are using periodic boundary conditions; that is, $\sigma_1 = \sigma_{N+1}$). For a detailed derivation, see Theorem \ref{thm:t-matrix} in Appendix \ref{app:c}. Here we adopt the convention $\prod_{i = N}^1 a_i = a_N\cdot a_{N-1}\cdots a_1$. Note that, unlike in the periodic case, the matrices $\set{\mathcal{T}_{\omega,i}}$ do not commute and so cannot be diagonalized simultaneously. This is precisely what complicates the analysis, as the partition function, and hence the free energy function, cannot be computed explicitly.  

For a matrix $M = (a_{ij})$, define the one-norm and the infinity-norm, respectively, by
\begin{align*}
\norm{M}_1 := \sum_{i,j}\abs{a_{ij}}\hspace{2mm}\text{ and }\hspace{2mm}\norm{M}_\infty := \max_{i,j}\set{\abs{a_{ij}}}.
\end{align*}
From \eqref{eq:mat-form} it is evident that for each $\mathcal{T}_{\omega, i}$, $\norm{\mathcal{T}_{\omega,i}}_\infty = [\mathcal{T}_{\omega,i}]_{11}$. Hence from \eqref{eq:partition-trace-form} we obtain
\begin{align}\label{eq:part-lyap}
\norm{\mathcal{T}^{(N)}_\omega}_\infty\leq \mathbf{Z}^{(N)}_\omega\leq \norm{\mathcal{T}^{(N)}_\omega}_1.
\end{align}
From \eqref{eq:part-lyap} and \eqref{eq:eq-real-free-energy} we obtain 
\begin{align}\label{eq:free-lyap}
-\frac{1}{Nk_B\tau}\log\norm{\mathcal{T}^{(N)}_\omega}_\infty\geq \mathbf{F}^{(N)}_\omega \geq -\frac{1}{Nk_B\tau}\log\norm{\mathcal{T}^{(N)}_\omega}_1.
\end{align}
We can now apply Theorem \ref{thm:walters} with $G: \Omega\rightarrow \mathrm{GL}(2, \R)$ given by $G(\omega) = \mathcal{T}_{\omega, 1}$, in which case $\mathcal{T}_\omega^{(N)} = G(T^{N-1}\omega)G(T^{N-2}\omega)\cdots G(\omega)$ (recall that $T:\Omega\rightarrow\Omega$ is the left shift). That $G$ is continuous is obvious, since $G$ depends only on the first element of the sequence $\omega$; also every entry of $G(\omega)$ is strictly positive, as is evident from \eqref{eq:mat-form}. Now applying Theorem \ref{thm:walters} on both sides of the inequality in \eqref{eq:free-lyap}, we see that
\begin{align*}
\mathbf{F}:=\lim_{N\rightarrow\infty}\mathbf{F}^{(N)}_\omega
\end{align*}
exists as a function of temperature $\tau\in(0,\infty)$ and is independent of $\omega$. From \eqref{eq:part-lyap} and \eqref{eq:free-lyap} it follows that $\mathbf{F}$ is strictly negative. It remains to prove that $\mathbf{F}$ is analytic in $\tau$.

Since the limit in \eqref{eq:free-lyap} is independent of the choice $\omega\in\Omega$, we may choose $\omega_s\in\Omega$, where $\omega_s$ is the special sequence as constructed in Section \ref{subsec:fibonacci-substitution}. From now on we drop $\omega$ in all formulas, assuming implicitly that $\omega = \omega_s$. 

Since the limit in \eqref{eq:free-lyap} exists as $N\rightarrow\infty$, we may pass to a subsequence $N_k = F_k$, where $F_k$ is the $k$th Fibonacci number, and by \eqref{eq:partition-trace-form} it is enough to prove analyticity of the limit
\begin{align}\label{eq:trace-lim}
\lim_{k\rightarrow\infty}\frac{\log \Tr\left(\mathcal{T}^{(F_k)}\right)}{F_k} = \lim_{k\rightarrow\infty}\frac{\log \Tr\left(\mathcal{T}^{(F_k)}\right)}{\sqrt{5}\phi^k},
\end{align}
where $\phi$ is the golden mean (the last equality follows since $\lim_{k\rightarrow\infty}F_k/\phi^k = 1/\sqrt{5}$). 

Let us define
\begin{align*}
\widetilde{\mathcal{T}}_i := \frac{1}{d_i}\mathcal{T}_i,\hspace{2mm}\text{ where }\hspace{2mm} d_i := \sqrt{\det \mathcal{T}_i},\hspace{2mm}\text{ and set }\hspace{2mm} x_k := \frac{1}{2}\Tr\left(\prod_{i = F_k}^1\widetilde{\mathcal{T}}_i\right).
\end{align*}
Since $\widetilde{\mathcal{T}}_i$ is unimodular, the trace map can now be applied to these matrices. For $k \geq 3$, we have (see, for example, \cite{Baake1999} and references therein)
\begin{align}\label{eq:lim-analytic-0}
\Tr\left(\mathcal{T}^{(F_k)}\right) = \left(2\prod_{i = 1}^{F_k}d_i\right)\pi\circ f^{k-3}(x_3, x_2, x_1),
\end{align}
where $\pi$ is projection onto the first coordinate, and $f$ is the Fibonacci trace map as defined in \eqref{eq:trace-map}. 

Observe that the sequence $\set{d_i}_{i\in\N}$ is two-valued modulated by the sequence $u$, as defined in Section \ref{subsec:fibonacci-substitution}. We have
\begin{align*}
\lim_{k\rightarrow\infty}\frac{\log \left(2\prod_{i = 1}^{F_k} d_i\right)}{F_k} = \lim_{k\rightarrow\infty}\frac{\log\left(2d_1^{F_{k-1}}d_2^{F_{k-2}}\right)}{F_k} = \frac{1}{\phi}\log d_1 + \frac{1}{1 + \phi}\log d_2,
\end{align*}
which is obviously analytic as a function of $\tau$. Thus in order to complete the proof, we need to prove that
\begin{align}\label{eq:lim-analytic-2}
\lim_{k\rightarrow\infty}\frac{\log \pi\circ f^{k-3}(x_3, x_2, x_1)}{F_k}\hspace{2mm}\text{ exists and is analytic in } \tau.
\end{align}

The initial condition $(x_3, x_2, x_1)$ defines an analytic curve $(0,\infty)\rightarrow\R^3$ as a map of temperature, $\tau$. Since the limit in \eqref{eq:lim-analytic-2} exists and is nonzero (by \eqref{eq:free-lyap} and application of Theorem \ref{thm:walters}), $\set{\pi\circ f^{k-3}(x_3,x_2,x_1)}_{k\geq 3}$ is unbounded, hence $(x_3, x_2, x_1)$ escapes under the action of $f$.

Fix $\tau_0\in (0, \infty)$. Consider the complexification of $(x_3, x_2, x_1)$. Since for all $\tau\in (0, \infty)$, $(x_3, x_2, x_1)$ escapes under the action of $f$, and the set of points in $\C^3$ which escape under $f$ is an open set (see \cite{Cantat2009}), there exists a small neighborhood $U$ around $(x_3(\tau_0), x_2(\tau_0), x_1(\tau_0))$ in $\C^3$, such that all $z\in U$ have unbounded forward semi-orbit. Application of Proposition \ref{prop:escape-rate} gives \eqref{eq:lim-analytic-2}. 

%




\subsection{Proof of Theorem \ref{thm:complex-ising-zeros}}\label{subsec:proof-thm-complex-ising-zeros}

Notice that the proof of the previous theorem gives a recursion relation for the partition function itself. Indeed, for $k\geq 3$,
\begin{align*}
\mathbf{Z}^{(F_k)} = \left(2\prod_{i = 1}^N d_i\right)\pi\circ f^{k-3}(x_3, x_2, x_1).
\end{align*}
Since we are interested in the zeros of $\mathbf{Z}$, it is enough to consider
\begin{align*}
\widetilde{\mathbf{Z}}^{(F_k)}:= \pi\circ f^{k-3}(x_3, x_2, x_1)
\end{align*}
as a function of the variable $\eta$, as defined in \eqref{eq:complex-regime-variables}. In the new variables $\alpha$, $\beta$ and $\eta$ we have
\begin{align*}
f^{-1}(x_3, x_2, x_1) = \left(\frac{\alpha\beta\widetilde{\eta}}{2\sqrt{\beta^4 - 1}}, \frac{\alpha\beta\widetilde{\eta}}{2\sqrt{\alpha^4 - 1}}, \frac{\alpha^2\beta^2 - 1}{\sqrt{(\alpha^4 - 1)(\beta^4 - 1)}}\right),
\end{align*}
where $\widetilde{\eta} = \eta + \overline{\eta}$. We took $f^{-1}$ to obtain a simpler expression. Since the zeros of $\widetilde{\mathbf{Z}}^{(F_k)}$ all lie on $S^1$, we restrict the allowed values for $\widetilde{\eta}$ to $[-2,2]$.

Set
\begin{align}\label{eq:the-line}
\gamma(\widetilde{\eta}) = \left(\frac{\alpha\beta\widetilde{\eta}}{2\sqrt{\beta^4 - 1}}, \frac{\alpha\beta\widetilde{\eta}}{2\sqrt{\alpha^4 - 1}}, \frac{\alpha^2\beta^2 - 1}{\sqrt{(\alpha^4 - 1)(\beta^4 - 1)}}\right).
\end{align}
Notice that $\gamma$ is an analytic line in $\R^3$. For $k\geq 3$, let
\begin{align*}
\mathcal{Z}_k = \set{\widetilde{\eta}\in[-2,2]: \pi\circ f^{k-2}(\gamma(\widetilde{\eta})) = 0}.
\end{align*}
We are interested to know whether $\set{\mathcal{Z}_k}$ converges in Hausdorff metric to some compact (nonempty) $\mathcal{Z}\subset [-2,2]$, and if so, what is the topological, measure-theoretic and fractal-dimensional structure of $\mathcal{Z}$. 

Define 
\begin{align*}
\mathcal{Z} = \set{\widetilde{\eta}\in[-2,2]: \mathcal{O}^+_f(\gamma(\widetilde{\eta}))\text{ is bounded}}.
\end{align*}
Observe that $\mathcal{Z}$ may contain finitely many isolated points (at such isolated points $\gamma$ would intersect the center-stable manifolds tangentially, and we know from Theorem \ref{thm:thm-yessen} that there are at most finitely many tangential intersections); let us call this set of isolated points $\mathcal{T}$. We now prove that $\mathcal{Z}_k\xrightarrow[k\rightarrow\infty]{}\mathcal{Z}\setminus{\widetilde{\mathcal{T}}}$ in Hausdorff metric, where $\widetilde{\mathcal{T}}\subseteq \mathcal{T}$ (we believe that $\mathcal{T} = \emptyset$, but we are currently unable to prove this). 

For the sake of convenience, define 
\begin{align}\label{eq:tilde-zet}
\widetilde{\mathcal{Z}}:=\mathcal{Z}\setminus\widetilde{\mathcal{T}}.
\end{align}
We also prove that $\widetilde{\mathcal{Z}}$ satisfies i--iii of Theorem \ref{thm:complex-ising-zeros}, which will complete the proof. Let us first prove that $\widetilde{\mathcal{Z}}$ satisfies i--iii of Theorem \ref{thm:complex-ising-zeros}.

The value of the Fricke-Vogt invariant along $\gamma(\widetilde{\eta})$ is given by
\begin{align}\label{eq:fv-along-curve}
I(\gamma(\widetilde{\eta})) = \frac{(\alpha^2\beta^2\widetilde{\eta}^2 + 4)(\alpha^2 - \beta^2)^2}{4(\alpha^4 - 1)(\beta^4 - 1)}.
\end{align}
Observe that $I(\gamma(\widetilde{\eta})) = 0$ if and only if $\alpha = \beta$, in which case $\gamma$ lies entirely on $S_0^\R$. This corresponds to the case $p(a) = p(b)$, i.e. the pure case in which the partition function and corresponding zeros can be computed explicitly. Otherwise, $I(\gamma(\widetilde{\eta}))>0$, i.e. $\gamma$ lies entirely in $\bigcup_{V > 0}S_V^\R$. Moreover, $I(\gamma(\widetilde{\eta}))$ is non-constant as a function of $\widetilde{\eta}$. We are almost in the position where we could apply Theorem \ref{thm:thm-yessen}.

Next, let us show that $\gamma$ cannot lie entirely on a center-stable manifold. Indeed, let us show that the endpoints of $\gamma$, that is, $\gamma(-2)$ and $\gamma(2)$, escape to infinity. Indeed, at $\pm 2$ we have
\begin{align*}
(x,y,z):=\gamma(\pm 2) = \left(\frac{\pm\alpha\beta}{\sqrt{\beta^4 - 1}}, \frac{\pm\alpha\beta}{\sqrt{\alpha^4 - 1}}, \frac{\alpha^2\beta^2 - 1}{\sqrt{(\alpha^4 - 1)(\beta^4 - 1)}}\right).
\end{align*}
Clearly either $\abs{x} > 1$ or $\abs{y} > 1$. Using this, a direct computation shows that either $f(x,y,z)$ or $f^{-1}(x,y,z)$ satisfies that the first two coordinates in absolute value are greater than one, and their product is greater than the third coordinate. It follows that these points escape to infinity (see the discussion in Section \ref{subsec:proof-thm-spectral-convergence} following equation \eqref{eq:gamma-cc-intersection}). This also proves ii of Theorem \ref{thm:complex-ising-zeros}, since $\widetilde{\eta} = \pm 2$ corresponds to $\eta = \pm 1$. 

Next, observe that
\begin{align}\label{eq:check-point}
\gamma(0) = \left(0, 0, \frac{\alpha^2\beta^2 - 1}{\sqrt{(\alpha^4 - 1)(\beta^4 - 1)}}\right)
\end{align}
is a periodic point of period six (the exact form of this expression is not important; in general points of the form where two of the coordinates are zero are periodic). It follows that $\gamma$ intersects center-stable manifolds at points other than the endpoints. 

Finally, let us show that the intersection of $\gamma$ with the center-stable manifold containing the point $\gamma(0)$ is transversal. For convenience, let us call this center-stable manifold $W$. From \eqref{eq:fv-along-curve} we have
\begin{align*}
I(\gamma(0)) > 0\hspace{2mm}\text{ and }\left.\frac{\partial I}{\partial \widetilde{\eta}}\right|_{\gamma(0)} = 0.
\end{align*}
Hence $\gamma(0)\in S_V^\R$ with $V > 0$ and $\gamma'(0)\in T_{\gamma(0)}S_{V}^\R$. On the other hand, intersections of center-stable manifolds with surfaces $S_V^\R$ for any $V > 0$ are transversal (see Proposition \ref{prop:phyperb} above). Therefore we need to prove that $\gamma'(0)$ is transversal in $T_{\gamma(0)}S_V^\R$ to the tangent space at $\gamma(0)$ of the smooth curve $W\cap S_V^\R$. On the other hand, the curve $W\cap S_V^\R$ is precisely the stable manifold to the hyperbolic periodic point $\gamma(0)\in S_V^\R$. Since this periodic point is of period six, passing to $f^6$, we have $f^6(\gamma(0)) = \gamma(0)$ -- a fixed hyperbolic point under $f^6$. Hence the tangent space at $\gamma(0)$ to $W\cap S_V^\R$ coincides with the eigenspace corresponding to the smallest (in absolute value) eigenvalue of $Df^6_{\gamma(0)}$---since $f$ is easily seen to be area preserving, at a hyperbolic fixed point the differential has three distinct eigenvalues: $\set{1, \lambda, \alpha/\lambda}$, where $\abs{\
lambda} > 1$ and $\alpha\in\set{\pm 1}$. As can easily be checked, the sought eigenvalue and the corresponding eigenvector are
\begin{align*}
-4a^2\sqrt{4a^4 + 1} + 8a^4 + 1\hspace{2mm}\text{ and }\hspace{2mm}\left(1,-\frac{\sqrt{4a^4 + 1} + 2a^2 - 1}{2a}, 0\right),
\end{align*}
with $a = \frac{\alpha^2\beta^2 - 1}{\sqrt{(\alpha^4 - 1)(\beta^4 - 1)}} > 0$. On the other hand, as is evident from \eqref{eq:the-line}, $\gamma'$ is spanned by the vector $(1,\sqrt{(\beta^4 - 1)(\alpha^4 - 1)^{-1}}, 0)$. Since the second entry in this vector is strictly positive, while the second entry in the eigenvector is strictly negative, and all other entries coincide, the two vectors are clearly transversal in $T_{\gamma(0)}S_V^\R$, and we're done.

Now Theorem \ref{thm:thm-yessen} can be applied and i of Theorem \ref{thm:complex-ising-zeros} follows. 

Next let us prove that $\widetilde{\mathcal{Z}}$ satisfies iii of Theorem \ref{thm:complex-ising-zeros}. Indeed, by Theorem \ref{thm:thm-yessen} it is enough to prove that for all $p(a)$ and $p(b)$ sufficiently close (i.e. for all $\alpha$ and $\beta$ sufficiently close), $\gamma$ in \eqref{eq:the-line} intersects the center-stable manifolds transversally. This follows from the same arguments as in \cite[Proof of Theorem 2.5]{Yessen2011a}. 

Finally, let us prove that $\mathcal{Z}_k\xrightarrow[k\rightarrow\infty]{}\widetilde{\mathcal{Z}}$ in Hausdorff metric. In place of $f$, let us consider $g := f^6$ (the sixths iterate of $f$). That is, let us show that $\mathcal{Z}_{6k}\xrightarrow[k\rightarrow\infty]{}\widetilde{\mathcal{Z}}$ in Hausdorff metric. This is done merely for the sake of technical convenience, as is apparent from the proof below, and does not result in loss of generality.

Fix arbitrary $\epsilon > 0$ and let $U_1, \dots, U_{m_\epsilon}$ be an open cover of $\mathcal{Z}$ with the diameter $\diam(U_i)\leq \epsilon$. It is enough to show that there exists $K_\epsilon \in \N$ such that for all $k\geq K_\epsilon$, $\mathcal{Z}_{6k}$ is contained in $\bigcup_{i=1}^{m_\epsilon}U_i$, and for every $i\in \set{1,\dots,m_\epsilon}$, $\mathcal{Z}_{6k}\cap U_i\neq\emptyset$ if $U_i\cap \mathcal{Z}$ is not a finite set (i.e., if $U_i$ does not only contain isolated points of $\mathcal{Z}$). 

Every point of $\gamma\cap\left(\bigcup_{i = 1}^{m_\epsilon}U_i\right)^c$ escapes to infinity. Hence there exists $K_1\in\N$ such that for all $k\geq K_1$, $\mathcal{Z}_{6k}\subset \bigcup_{i = 1}^{m_\epsilon}U_i$. 

Consider the four singularities
\begin{align*}
P_1 = (1,1,1),\hspace{2mm} P_2 = (-1,-1,1),\hspace{2mm} P_3 = (1,-1,-1),\hspace{2mm} P_4 = (-1,1,-1)
\end{align*}
of the Cayley cubic $S_0^\R$. The first is fixed under $f$, and the other three form a three-cycle. When the parameter $V$ is turned on, the first singularity bifurcates into a hyperbolic periodic point of period two. The three cycle bifurcates into a hyperbolic six cycle. More precisely, there are four smooth curves in $\R^3$ that lie entirely in $\bigcup_{V\geq 0}S_V^\R$, one through each $P_i$, $i = 1, \dots, 4$. Label these curves $\rho_i$. Each $\rho_i$ is divided into two smooth branches by $P_i$; label these branches $\rho_i^l$ and $\rho_i^r$. Under $f$, $\rho_1^l\mapsto \rho_1^r$, while the other six branches are cycled through a six-cycle. Points on these curves are hyperbolic. In particular, under $g$ (recall: $g = f^6$) these curves are fixed normally hyperbolic submanifolds of $\R^3$ (for a brief discussion of normal hyperbolicity, see Appendix \ref{b3}, and references therein for details). Each of $\rho_i^{l,r}$, $i = 1, \dots, 4$, intersects each $S_V^\R$, $V > 0$, transversally and in only one 
point. This point of intersection is a fixed hyperbolic point for $g$ (for more details see Appendix \ref{b}). For further details on curves on six cycles, see Appendix \ref{a1}.

The stable manifolds to each of the six branches are members of the family of center-stable manifolds and are dense in this family (for definitions, see Appendix \ref{b3} below, and references therein). Moreover, among the six branches, there are those that lie below the plane $\set{z = 0}$, and those that lie above (for details, see the discussion on symmetries of the Fibonacci trace map in Appendix \ref{a2}). Consequently, for each $i_0\in \set{1, \dots, m_\epsilon}$, if $U_{i_0}$ does not contain only isolated points of $\mathcal{Z}$, then there is a pair of points $p, q\in \mathcal{Z}\cap U_{i_0}$, such that $p$ lies on the stable manifold to a branch that lies below $\set{z = 0}$, and $q$ lies on the stable manifold to a branch that lies above $\set{z = 0}$. Hence for all $k$ sufficiently large, $g^k(p)$ lies below $\set{z = 0}$, and $g^k(q)$ lies above $\set{z = 0}$. Therefore, if $\alpha$ denotes the arc along $\gamma$ that connects $p$ and $q$ (and \textit{a priori} lies in $U_{i_0}$), then for all 
such $k$, $g(\alpha)$ intersects $\set{z = 0}$. Hence for all such $k$, $\mathcal{Z}_{6k}\cap U_{i_0}\neq \emptyset$. Consequently, there exists $K_2\in\N$ such that for all $k\geq K_2$, $\mathcal{Z}_{6k}\cap U_i\neq \emptyset$ for every $i\in\set{1,\dots, m_\epsilon}$ such that $U_i$ does not contain only isolated points of $\mathcal{Z}$.

Let $K_\epsilon = K_1 + K_2$. This completes the proof.

\section*{Acknowledgment}

It is a pleasure to thank my dissertation advisor, Anton Gorodetski, for many helpful discussions and for a period of financial support during which this project was conceived and completed.

I would also like to thank the anonymous referees for their comments and suggestions, which helped improve both, the content and the presentation of this paper.

\appendixpage

\appendix

None of what is presented in the following appendices is new. In particular, all of appendix \ref{b} is by now part of the classical theory of hyperbolic and partially hyperbolic dynamical systems (and references are given to comprehensive surveys). 

\section{Normal hyperbolicity of six-cycles through singularities of \texorpdfstring{$S_0$}{} and symmetries of \texorpdfstring{$f$}{}}\label{a}

\subsection{Normal hyperbolicity of six-cycles through singularities of \texorpdfstring{$S_0$}{}}\label{a1}

As has been mentioned above, $S_0$ contains four conic singularities; explicitly, they are 
\begin{align}\label{eq:singularities} 
P_1 = (1,1,1),\hspace{2mm} P_2 = (-1,-1,1),\hspace{2mm} P_3 = (1,-1,-1),\hspace{2mm} P_4 = (-1,1,-1). 
\end{align} 
The point $P_1$ is fixed under $f$, while $P_2$, $P_3$ and $P_4$ form a three cycle: 
\begin{align*} 
P_1\overset{f}{\longmapsto}P_1;\hspace{4mm} P_2\overset{f}{\longmapsto}P_3\overset{f}{\longmapsto}P_4\overset{f}{\longmapsto}P_2 
\end{align*} 
(which can be verified via direct computation).

For each $i\in\set{1,\dots,4}$, there is a smooth curve $\rho_i$ which does not contain any self-intersections, passing through the singularity $P_i$, such that $\rho_i\setminus{P_i}$ is a disjoint union of two smooth curves---call them $\rho_i^l$ and $\rho_i^r$---with the following properties: 
\begin{itemize} 

\item $\rho_i^{l,r}\subset \bigcup_{V>0}S_V^\R$; 

\item $f(\rho_1^l) = \rho_1^r$ and $f(\rho_1^r) = f(\rho_1^l)$. In particular, points of $\rho_1^{l,r}$ are periodic of period two, and $\rho_1$ is fixed under $f$; 

\item The six curves $\rho_i^{l,r}$, $i = 2, 3, 4$, form a six cycle under $f$. In particular, points of $\rho_i^{l,r}$, $i = 2, 3, 4$, are periodic of period six, and hence for $i = 2, 3, 4$, $\rho_i$ is fixed under $f^6$. 

\end{itemize}
The curve $\rho_1$ is given explicitly by
\begin{align}\label{eq:per-2-curve} 
\rho_1 = \set{\left(x,\hspace{1mm}\frac{x}{2x-1},\hspace{1mm}x\right): x\in\left(-\infty,1/2\right)\cup\left(1/2,\infty\right)}. 
\end{align} 
Expressions for the other three curves can be obtained from \eqref{eq:per-2-curve} using symmetries of $f$ to be discussed below.  

It follows via a simple computation that for any $i \in \set{1,\dots, 4}$ and any point $p\in \rho_i^{l,r}$, the eigenvalue spectrum of $Df^6_p$ is $\set{1, \lambda(p), 1/\lambda(p)}$ with $0 < \abs{\lambda(p)} < 1$, where $Df^6_p$ denotes the differential of $f^6$ at the point $p$. The eigenspace corresponding to the eigenvalue $1$ is tangent to $\rho_i$ at $p$. At $P_i$, the eigenvalue spectrum of $Df^6_{P_i}$ is of the same form, and as above, the eigenspace corresponding to the unit eigenvalue is tangent to $\rho_i$ at $P_i$. It follows that the curves $\set{\rho_i}_{i\in\set{1,\dots,4}}$ are normally hyperbolic one-dimensional submanifolds of $\R^3$, as defined in Section \ref{b3}.

For $V > 0$, $S_V^\R\bigcap (\rho_1\cup\cdots\cup\rho_4)$ consists of eight hyperbolic periodic points for $f$ which are fixed hyperbolic points for $f^6$. The stable manifolds to these points form a dense sublamination of the stable lamination
\begin{align*}
W^s(\Omega_V) = \bigcup_{x\in\Omega_V}W^s(x),
\end{align*}
the union of global stable manifolds to points in $\Omega_V$, the nonwandering set for $f$ on $S_V$ from Theorem \ref{thm:cas-gor-cant}. Moreover, this sublamination forms the boundary of the stable lamination. For details, see \cite{Casdagli1986, Damanik2009, Cantat2009}. This extends to the center-stable lamination: the stable manifolds to the normally hyperbolic curves $\rho_1,\dots,\rho_4$ form a dense sublamination of the (two-dimensional) center-stable lamination and forms the boundary of this lamination. In particular, if $\gamma$ is a smooth curve intersecting a center-stable manifold, and this intersection is not isolated, then in an arbitrarily small neighborhood of this intersection, $\gamma$ intersects the stable manifolds of all eight curves: $\set{\rho_i^{l,r}}_{i\in\set{1,\dots,4}}$.

\subsection{Symmetries of \texorpdfstring{$f$}{}}\label{a2}

The following discussion is taken from \cite{Damanik0000my1}; however, what follows does not appear in \cite{Damanik0000my1} as new results but a recollection of what is known. In particular, the reader should consult \cite{Baake1997, Roberts1994} and references therein, as well as earlier (and original) works \cite{Kohmoto1983, Kohmoto1992, Kadanoff0000, Kadanoff1984}.

Let us denote the group of symmetries of $f^6$ by $\mathcal{G}_\mathrm{sym}$, and the group of reversing symmetries of $f^6$ by $\mathcal{G}_\mathrm{rev}$; that is, 
\begin{align}\label{eq:sym-group} 
\mathcal{G}_\mathrm{sym} = \set{s\in\mathrm{Diff}(\R^3): s\circ f^6\circ s^{-1} = f^6}, 
\end{align} 
and 
\begin{align}\label{eq:rev-group} 
\mathcal{G}_\mathrm{rev} = \set{s\in\mathrm{Diff}(\R^3): s\circ f^6\circ s^{-1} = f^{-6}}, 
\end{align} 
where $\mathrm{Diff}(\R^3)$ denotes the set of diffeomorphisms on $\R^3$. 

Observe that $\mathcal{G}_\mathrm{rev}\neq\emptyset$. Indeed,  
\begin{align}\label{eq:rev-sym} 
s(x,y,z) = (z,y,x) 
\end{align} 
is a reversing symmetry of $f$, and hence also of $f^6$. Hence $f^6$ is smoothly conjugate to $f^{-6}$. It follows (see Appendix \ref{a1}) that forward-time dynamical properties of $f^6$, as well as the geometry of dynamical invariants (such as stable manifolds) are mapped smoothly and rigidly to those of $f^{-6}$. That is, forward-time dynamics of $f^6$ is essentially the same as its backward-time dynamics. 

The group $\mathcal{G}_\mathrm{sym}$ is also nonempty, and more importantly, it contains the following diffeomorphisms: 
\begin{align}\label{eq:symmetries} 
s_2: (x,y,z)\mapsto (-x, -y, z),\notag\\ 
s_3: (x,y,z)\mapsto(x,-y,-z),\\ 
s_4: (x,y,z)\mapsto(-x,y,-z).\notag 
\end{align} 
Notice that $s_i$ are rigid transformations. Also notice that 
\begin{align*}
s_i(P_1) = P_i,
\end{align*} 
and since $s_i$ is in fact a smooth conjugacy, we must have
\begin{align}\label{eq:symmetries-on-rho}
s_i(\rho_1) = \rho_i.
\end{align}
For a more general and extensive discussion of symmetries and reversing symmetries of trace maps, see \cite{Baake1997}.

\section{Background on uniform, partial and normal hyperbolicity}\label{b}

\subsection{Properties of locally maximal hyperbolic sets}\label{b1}

A more detailed discussion can be found in \cite{Hirsch1968, Hirsch1970, Hirsch1977, Hasselblatt2002b, Hasselblatt2002}.

A closed invariant set $\Lambda\subset M$ of a diffeomorphism $f: M\rightarrow M$ of a smooth manifold $M$ is called \textit{hyperbolic} if for each $x\in\Lambda$, there exists the splitting $T_x\Lambda = E_x^s\oplus E_x^u$ invariant under the differential $Df$, and $Df$ exponentially contracts vectors in $E_x^s$ and exponentially expands vectors in $E_x^u$. The set $\Lambda$ is called \textit{locally maximal} if there exists a neighborhood $U$ of $\Lambda$ such that
\begin{align}\label{part2_eq1}
\Lambda = \bigcap_{n\in\Z}f^n(U).
\end{align}
The set $\Lambda$ is called \textit{transitive} if it contains a dense orbit. It isn't hard to prove that the splitting $E_x^s\oplus E_x^u$ depends continuously on $x\in\Lambda$, hence $\dim(E_x^{s,u})$ is locally constant. If $\Lambda$ is transitive, then $\dim(E_x^{s,u})$ is constant on $\Lambda$. We call the splitting $E_x^s\oplus E_x^u$ a $(k_x^s, k_x^u)$ splitting if $\dim(E_x^{s,u}) = k^{s,u}$, respectively. In case $\Lambda$ is transitive, we shall simply write $(k^s, k^u)$. 
\begin{defn}\label{basic_set}
We call $\Lambda\subset M$ a \textit{basic set} for $f\in\diff^r(M)$, $r\geq 1$, if $\Lambda$ is a locally maximal invariant transitive hyperbolic set for $f$.
\end{defn}
Suppose $\Lambda$ is a basic set for $f$ with $(1,1)$ splitting. Then the following holds.

\subsubsection{Stability}\label{b1-1}

Let $U$ be as in \eqref{part2_eq1}. Then there exists $\mathcal{U}\subset \diff^1(M)$ open, containing $f$, such that for all $g\in\mathcal{U}$,
\begin{align}\label{part2_eq2}
\Lambda_g = \bigcap_{n\in\Z}g^n(U)
\end{align}
is $g$-invariant transitive hyperbolic set; moreover, there exists a (unique) homeomorphism $H_g:\Lambda\rightarrow\Lambda_g$ such that
\begin{align}\label{part2_eq3}
H_g\circ f|_{\Lambda} = g|_{\Lambda_g}\circ H_g;
\end{align}
that is, the following diagram commutes.
\begin{center}
\begin{tikzpicture}
 \matrix (m) [matrix of math nodes,
	      row sep = 3em,
	      column sep = 6em,
	      minimum width = 2em,text height=1.5ex, text depth=0.25ex]
 {
  \Lambda	&	\Lambda\\
  \Lambda_g	&	\Lambda_g\\
 };
 \path[-stealth]
  (m-1-1) edge node [above] {$f$} (m-1-2)
  (m-1-2) edge node [right] {$H_g$} (m-2-2)
  (m-1-1) edge node [left]  {$H_g$} (m-2-1)
  (m-2-1) edge node [below] {$g$} (m-2-2);
\end{tikzpicture}
\end{center}
Also $H_g$ can be taken arbitrarily close to the identity by taking $\mathcal{U}$ sufficiently small. In this case $g$ is said to be \textit{conjugate to} $f$, and $H_g$ is said to be \textit{the conjugacy}.

\subsubsection{Stable and unstable invariant manifolds}\label{b1-2}

Let $\epsilon > 0$ be small. For each $x\in\Lambda$ define the \textit{local stable} and \textit{local unstable} manifolds at $x$:
\begin{align*}
W_\epsilon^s(x) = \set{y\in M: d(f^n(x),f^n(y))\leq \epsilon \text{ for all }n\geq 0},
\end{align*}
\begin{align*}
W_\epsilon^u(x) = \set{y\in M: d(f^n(x),f^n(y))\leq \epsilon \text{ for all }n\leq 0}.
\end{align*}
We sometimes do not specify $\epsilon$ and write
\begin{align*}
W_\loc^s(x)\text{\hspace{5mm}and\hspace{5mm}}W_\loc^u(x)
\end{align*}
for $W_\epsilon^s(x)$ and $W_\epsilon^u(x)$, respectively, for (unspecified) small enough $\epsilon > 0$. For all $x\in\Lambda$, $W_\loc^{s,u}(x)$ is an embedded $C^r$ disc with $T_xW_\loc^{s,u}(x) = E_x^{s,u}$. The \textit{global stable} and \textit{global unstable} manifolds
\begin{align}\label{part2_eq4}
W^s(x) = \bigcup_{n\in\N}f^{-n}(W_\loc^s(x))\text{\hspace{5mm}and\hspace{5mm}}W^u(x) = \bigcup_{n\in\N}f^{n}(W_\loc^u(x))
\end{align}
are injectively immersed $C^r$ submanifolds of $M$. Define also the stable and unstable sets of $\Lambda$:
\begin{align}\label{part2_eq5}
W^s(\Lambda) = \bigcup_{x\in\Lambda}W^s(x)\text{\hspace{5mm}and\hspace{5mm}}W^u(\Lambda) = \bigcup_{x\in\Lambda}W^u(x).
\end{align}

If $\Lambda$ is compact, there exists $\epsilon > 0$ such that for any $x,y\in\Lambda$, $W_\epsilon^s(x)\cap W_\epsilon^u(y)$ consists of at most one point, and there exists $\delta > 0$ such that whenever $d(x,y) < \delta$, $x,y\in\Lambda$, then $W_\epsilon^s(x)\cap W_\epsilon^u(y)\neq \emptyset$. If in addition $\Lambda$ is locally maximal, then $W_\epsilon^s(x)\cap W_\epsilon^u(y)\in\Lambda$. 

The stable and unstable manifolds $W_\loc^{s,u}(x)$ depend continuously on $x$ in the sense that there exists $\Phi^{s,u}:\Lambda\rightarrow\emb^r(\R, M)$ continuous, with $\Phi^{s,u}(x)$ a neighborhood of $x$ along $W_\loc^{s,u}(x)$, where $\emb^r(\R, M)$ is the set of $C^r$ embeddings of $\R$ into $M$ \cite[Theorem 3.2]{Hirsch1968}.

The manifolds also depend continuously on the diffeomorphism in the following sense. For all $g\in\diff^r(M)$ $C^r$ close to $f$, define
$\Phi_g^{s,u}:\Lambda_g\rightarrow\emb^r(\R,M)$ as we defined $\Phi^{s,u}$ above. Then define
\begin{align*}
\tilde{\Phi}_g^{s,u}:\Lambda\rightarrow\emb^r(\R, M)
\end{align*}
by
\begin{align*}
\tilde{\Phi}_g^{s,u} = \Phi_g^{s,u}\circ H_g.
\end{align*}
Then $\tilde{\Phi}^{s,u}_g$ depends continuously on $g$ \cite[Theorem 7.4]{Hirsch1968}.

\subsubsection{Fundamental domains}\label{b1-3}

Along every stable and unstable manifold, one can construct the so-called \textit{fundamental domains} as follows. Let $W^s(x)$ be the stable manifold at $x$. Let $y\in W^s(x)$. We call the arc $\gamma$ along $W^s(x)$ with endpoints $y$ and $f^{-1}(y)$ a \textit{fundamental domain}. The following holds.
\begin{itemize}
\item $f(\gamma)\cap W^s(x) = y$ and $f^{-1}(\gamma)\cap W^s(x) = f^{-1}(y)$, and for any $k\in\Z$, if $k < -1$, then $f^k(\gamma)\cap W^s(x) = \emptyset$; if $k > 1$ then $f^k(\gamma)\cap W^s(x) = \emptyset$ iff $x\neq y$;
\item For any $z\in W^s(x)$, if for some $k\in\N$, $f^k(z)$ lies on the arc along $W^s(x)$ that connects $x$ and $y$, then there exists $n\in\N$, $n\leq k$, such that $f^n(z)\in\gamma$.
\end{itemize}
Similar results hold for the unstable manifolds.

\subsubsection{Invariant foliations}\label{a_1_3}

A stable foliation for $\Lambda$ is a foliation $\mathcal{F}^s$ of a neighborhood of $\Lambda$ such that
\begin{enumerate}
\item for each $x\in\Lambda$, $\mathcal{F}(x)$, the leaf containing $x$, is tangent to $E_x^s$;
\item for each $x$ sufficiently close to $\Lambda$, $f(\mathcal{F}^s(x))\subset\mathcal{F}^s(f(x))$.
\end{enumerate}
An unstable foliation $\mathcal{F}^u$ is defined similarly.

For a locally maximal hyperbolic set $\Lambda\subset M$ for $f\in\diff^1(M)$, $\dim(M) = 2$, stable and unstable $C^0$ foliations with $C^1$ leaves can be constructed; in case $f\in\diff^2(M)$, $C^1$ invariant foliations exist (see \cite[Section A.1]{Palis1993} and the references therein).

\subsubsection{Local Hausdorff and box-counting dimensions}\label{b1-4}

For $x\in\Lambda$ and $\epsilon > 0$, consider the set $W_\epsilon^{s,u}\cap\Lambda$. Its Hausdorff dimension is independent of $x\in\Lambda$ and $\epsilon > 0$.

Let
\begin{align}\label{part2_eq6}
h^{s,u}(\Lambda) = \dim_H(W_\epsilon^{s,u}(x)\cap\Lambda).
\end{align}
For properly chosen $\epsilon > 0$, the sets $W_\epsilon^{s,u}(x)\cap\Lambda$ are dynamically defined Cantor sets, so
\begin{align*}
h^{s,u}(\Lambda) < 1
\end{align*}
(see \cite[Chapter 4]{Palis1993}). Moreover, $h^{s,u}$ depends continuously on the diffeomorphism in the $C^1$-topology \cite{Manning1983}. In fact, when $\dim(M) = 2$, these are $C^{r-1}$ functions of $f\in\diff^r(M)$, for $r\geq 2$ \cite{Mane1990}.

Denote the box-counting dimension of a set $\Gamma$ by $\dim_{\mathrm{Box}}(\Gamma)$. Then
\begin{align*}
\dim_H(W^{s,u}_\epsilon(x)\cap \Lambda) = \dim_{\mathrm{Box}}(W^{s,u}_\epsilon(x)\cap \Lambda)
\end{align*}
(see \cite{Manning1983, Takens1988}).

\subsection{Partial hyperbolicity}\label{b2}

For a more detailed discussion, see \cite{Pesin2004, Hasselblatt2006}.

An invariant set $\Lambda\subset M$ of a diffeomorphism $f\in\diff^r(M)$, $r\geq 1$, is called \textit{partially hyperbolic (in the narrow sense)} if for each $x\in\Lambda$ there exists a splitting $T_xM = E_x^s\oplus E_x^c\oplus E_x^u$ invariant under $Df$, and $Df$ exponentially contracts vectors in $E_x^s$, exponentially expands vectors in $E_x^u$, and $Df$ may contract or expand vectors from $E_x^c$, but not as fast as in $E_x^{s,u}$. We call the splitting $(k_x^s, k_x^c, k_x^u)$ splitting if $\dim(E_x^{s,c,u}) = k_x^{s,c,u}$, respectively. We shall write $(k^s,k^c,k^u)$ if the dimension of subspaces does not depend on the point.

\subsection{Normal hyperbolicity}\label{b3}

For a more detailed discussion and proofs see \cite{Hirsch1977} and also \cite{Pesin2004}.

Let $M$ be a smooth Riemannian manifold, compact, connected and without boundary. Let $f\in\diff^r(M)$, $r\geq 1$. Let $N$ be a compact smooth submanifold of $M$, invariant under $f$. We call $f$ \textit{normally hyperbolic} on $N$ if $f$ is partially hyperbolic on $N$. That is, for each $x\in N$,
\begin{align*}
T_xM = E_x^s\oplus E_x^c\oplus E_x^u
\end{align*}
with $E_x^c = T_xN$. Here $E_x^{s,c,u}$ is as in Section \ref{b2}. Hence for each $x\in N$ one can construct local stable and unstable manifolds $W_\epsilon^s(x)$ and $W_\epsilon^u(x)$, respectively, such that
\begin{enumerate}
\item $x\in W_\loc^s(x)\cap W_\loc^u(x)$;
\item $T_x W_\loc^s(x) = E^s(x)$, $T_xW_\loc^u(x) = E^u(x)$;
\item for $n\geq 0$,
\begin{align*}
d(f^n(x),f^n(y))\xrightarrow[n\rightarrow\infty]{}0\text{ for all }y\in W_\loc^s(x),
\end{align*}
\begin{align*}
d(f^{-n}(x),f^{-n}(y))\xrightarrow[n\rightarrow\infty]{}0\text{ for all }y\in W_\loc^u(x).
\end{align*}
\end{enumerate}
(For the proof see \cite[Theorem 4.3]{Pesin2004}).
These can then be extended globally by
\begin{align}\label{eq:ss-manifold}
W^s(x)& = \bigcup_{n\in\N}f^{-n}(W^s_\loc(x));\\
W^u(x)& = \bigcup_{n\in\N}f^n(W^u_\loc(x)).
\end{align}

The manifold $W^s(x)$ is referred to as the \textit{strong-stable manifold}, while $W^u(x)$ is called the \textit{strong-unstable} manifold; sometimes to emphasize the point $x$, we add \textit{ at $x$}.

Set
\begin{align}\label{eq:cscu-def}
W_\loc^{cs}(N) = \bigcup_{x\in N}W_\loc^s(x)\text{\hspace{5mm}and\hspace{5mm}} W_\loc^{cu}(N) = \bigcup_{x\in N}W_\loc^u(x).
\end{align}
\begin{thm}[Hirsch, Pugh and Shub \cite{Hirsch1977}]\label{a1_thm1}
The sets $W_\loc^{cs}(N)$ and $W_\loc^{cu}(N)$, restricted to a neighborhood of $N$, are smooth submanifolds of $M$. Moreover,
\begin{enumerate}
\item $W^\mathrm{cs}_\loc(N)$ is $f$-invariant and $W^\mathrm{cu}_\loc$ is $f^{-1}$-invariant;
\item $N = W^\mathrm{cs}_\loc(N)\bigcap W^\mathrm{cu}_\loc(N)$;
\item For every $x\in N$, $T_x W^\mathrm{cs, cu}_\loc(N) = E_x^{s,u}\oplus T_xN$;
\item $W^\mathrm{cs}_\loc(N)$ ($W^\mathrm{cu}_\loc(N)$) is the only $f$-invariant ($f^{-1}$-invariant) set in a neighborhood of $N$;
\item $W^\mathrm{cs}_\loc(N)$ (respectively, $W^\mathrm{cu}_\loc(N)$) consists precisely of those points $y\in M$ such that for all $n\geq 0$ (respectively, $n\leq 0$), $d(f^n(x),f^n(y)) < \epsilon$ for some $\epsilon > 0$.
\item $W^\mathrm{cs,cu}_\loc(N)$ is foliated by $\set{W_\loc^\mathrm{s,u}(x)}_{x\in N}$.

\end{enumerate}

\end{thm}

\section{Derivation of transfer matrices}\label{app:c}

In this section we prove equality in equation \eqref{eq:partition-trace-form}. Let us start by introducing the following notation. We shall write
\begin{align*}
 (\underbrace{\pm \pm \cdots \pm}_{n \text{ times}} \singlebar \underbrace{\pm \pm \cdots \pm}_{m \text{ times}})\hspace{2mm}\text{for}\hspace{2mm}\exp(\pm K_1 \pm K_2 \cdots \pm K_n \pm h_1 \pm h_2 \cdots \pm h_m).
\end{align*}
For example, $(+ - + + \singlebar + - - +)$ denotes $\exp(K_1 - K_2 + K_3 + K_4 + h_1 - h_2 - h_3 + h_4)$. Thus the order of appearance of the signs in $(\cdots \singlebar \cdots)$ is important. We shall also write
\begin{align*}
 (\cdots \singlebar \cdots)\oplus(\cdots \singlebar \cdots)\hspace{2mm}\text{ for }\hspace{2mm}\exp(\cdots) + \exp(\cdots).
\end{align*}
For example, $(- + \singlebar + +)\oplus (- - \singlebar + -)$ stands for $\exp(-K_1 + K_2 + h_1 + h_2) + \exp(-K_1 - K_2 + h_1 - h_2)$. 

We define the transfer matrix at cite $i$ in the above notation by
\begin{align*}
 \mathfrak{T}_i = 
 \begin{pmatrix}[1.5]
  (+\singlebar+)_i & (-\singlebar-)_i\\
  (-\singlebar+)_i & (+\singlebar-)_i
 \end{pmatrix},
\end{align*}
where $(\pm \singlebar \pm)_i$ stands for $\exp(\pm K_i \pm h_i)$.

We also define multiplication $\star$ of $(\cdots \singlebar \cdots)$ and $(\circ\circ\circ\singlebar\circ\circ\circ)$ by
\begin{align*}
 (\cdots \singlebar \cdots)\star(\circ\circ\circ\singlebar\circ\circ\circ) = (\circ\circ\circ \cdots \singlebar \circ\circ\circ\cdots).
\end{align*}
Notice that multiplication is not in general commutative, precisely because the $(\cdots\singlebar\cdots)$ is not invariant under permutation of the signs $\cdots$.

In this notation, the partition function over the lattice of $N$ nodes is given by
\begin{align*}
 Z^{(N)} = \bigoplus_{u\in\set{+, -}^{N}}(\mathbf{P}(u)\singlebar u),
\end{align*}
where $\mathbf{P}(u)\in\set{+, -}^{N}$ is the \textit{allowed sequence} given the periodic boundary conditions and the configuration $(u_1, \dots, u_N)$ of spins on the lattice (in this notation $+$ stands for $+1$ and $-$ stands for $-1$; that is, spin up and down, respectively). 

In what follows, we shall prove a general result that is independent of the way that the coupling constants $K_i$, and the magnetic field $h_i$ are chosen (i.e. we make no use of quasi-periodicity or any other structure of the sequences). We deliberately use the letter $\mathfrak{T}$ in place of $\mathcal{T}$ as in \eqref{eq:partition-trace-form} for the transfer matrices to indicate the general nature of the result proved here.

Now using the definition of multiplication $\star$ from above, we can define multiplication $\otimes$ of two matrices with elements $\set{a_{ij}}$ and $\set{b_{ij}}$ of the form $(\cdots\singlebar\cdots)$ as
\begin{align*}
 \begin{pmatrix}
  a_{11} & \cdots & a_{1n}\\
  \vdots & \ddots & \vdots\\
  a_{m1} & \cdots & a_{mn}
 \end{pmatrix}
 \otimes
 \begin{pmatrix}
  b_{11} & \cdots & b_{1m}\\
  \vdots & \ddots & \vdots\\
  b_{n1} & \cdots & b_{nm}
 \end{pmatrix}
 =
 \begin{pmatrix}
  \overset{n}{\underset{i = 1}{\bigoplus}} a_{1i}\star b_{i1} & \cdots & \overset{n}{\underset{i = 1}{\bigoplus}} a_{1i}\star b_{im}\\
  \vdots & \ddots & \vdots\\
  \overset{n}{\underset{i = 1}{\bigoplus}} a_{ni}\star b_{im} & \cdots & \overset{n}{\underset{i = 1}{\bigoplus}} a_{mi}\star b_{im}
 \end{pmatrix}
\end{align*}

We now state and prove the main theorem of this section, relating the traces of transfer matrices to the partition function.

\begin{thm}\label{thm:t-matrix}
 The partition function over the lattice of size $N$ is given by
 \begin{align*}
  Z^{(N)} = \Tr(\mathfrak{T}_N\otimes\cdots\otimes \mathfrak{T}_1).
 \end{align*}
\end{thm}

To prove the theorem, we prove the following lemma, from which the theorem clearly follows. 

\begin{lem}\label{lem:t-matrix}
 For all $N\in\N$, the $N$-step transfer matrix is given by
 \begin{align*}
  \mathfrak{T}_N\otimes\cdots\otimes \mathfrak{T}_1 = 
  \begin{pmatrix}[2]
   \underset{\overset{+}{u}\in\set{+, -}^N}{\bigoplus}(\mathbf{P}(\overset{+}{u})\singlebar  \overset{+}{u}) & \underset{\overset{-}{u}\in\set{+, -}^N}{\bigoplus}(\mathbf{A}(\overset{-}{u})\singlebar\overset{-}{u})\\
   \underset{\overset{+}{u}\in\set{+, -}^N}{\bigoplus}(\mathbf{A}(\overset{+}{u})\singlebar \overset{+}{u}) & \underset{\overset{-}{u}\in\set{+, -}^N}{\bigoplus}(\mathbf{P}(\overset{-}{u})\singlebar \overset{-}{u}),
  \end{pmatrix}
 \end{align*}
 where $\overset{+}{u}$ is a sequence that begins with $+$ and $\overset{-}{u}$ is a sequence that begins with $-$; $\mathbf{A}(v)$ stands for the allowed sequence given anti-periodic boundary conditions with configuration $v$ (by anti-periodic we mean $\sigma_{N+1} = -\sigma_1$).
\end{lem}

The proof of the lemma is elementary and follows by induction. We only prove the result for the upper left corner element of the matrix $\mathfrak{T}_N\otimes\cdots\otimes \mathfrak{T}_1$, the proof for the other three positions being essentially the same.

\begin{proof}[Proof of Lemma \ref{lem:t-matrix}]
 Observe that the base case $N = 2$ holds, since
 \begin{align*}
  \mathfrak{T}_2\otimes\mathfrak{T}_1 =
  \begin{pmatrix}[1.5]
   (+ + \singlebar + +)\oplus (- - \singlebar + -) & (- + \singlebar - +)\oplus (+ - \singlebar - -)\\
   (+ - \singlebar + +)\oplus (- +\singlebar + -) & (- - \singlebar - +)\oplus (+ +\singlebar - -)
  \end{pmatrix},
 \end{align*}
 and it's a direct computation to verify that this matrix satisfies the form claimed in the lemma. Now let us suppose that for some $N\in\N$, the result holds. We have
 \begin{align*}
  \mathfrak{T}_{N+1}\otimes\cdots\otimes\mathfrak{T}_1
  = \mathfrak{T}_{N+1}
  \begin{pmatrix}[2]
   \underset{\overset{+}{u}\in\set{+, -}^N}{\bigoplus}(\mathbf{P}(\overset{+}{u})\singlebar  \overset{+}{u}) & \underset{\overset{-}{u}\in\set{+, -}^N}{\bigoplus}(\mathbf{A}(\overset{-}{u})\singlebar\overset{-}{u})\\
   \underset{\overset{+}{u}\in\set{+, -}^N}{\bigoplus}(\mathbf{A}(\overset{+}{u})\singlebar \overset{+}{u}) & \underset{\overset{-}{u}\in\set{+, -}^N}{\bigoplus}(\mathbf{P}(\overset{-}{u})\singlebar \overset{-}{u})
  \end{pmatrix},
 \end{align*}
 the upper left entry of which is given by
 \begin{align*}
  \left(\underset{\overset{+}{u}\in\set{+, -}^N}{\bigoplus}(+\singlebar +)_{N+1}\star(\mathbf{P}(\overset{+}{u})\singlebar \overset{+}{u})\right)&\oplus\left(\underset{\overset{+}{u}\in\set{+, -}^N}{\bigoplus}(-\singlebar -)_{N + 1}\star(\mathbf{A}(\overset{+}{u})\singlebar \overset{+}{u})\right)\\[1em]
  =
  \left(\underset{\overset{+}{u}\in\set{+, -}^N}{\bigoplus}(\mathbf{P}(\overset{+}{u}) \hspace{1mm} + \singlebar \overset{+}{u} \hspace{1mm} +)\right) &\oplus \left(\underset{\overset{+}{u}\in\set{+, -}^N}{\bigoplus}(\mathbf{A}(\overset{+}{u}) \hspace{1mm} - \singlebar \overset{+}{u} \hspace{1mm} -)\right)
 \end{align*}
 It isn't hard to see that the $\oplus$-sum above consists of exactly $2^N$ distinct terms corresponding to the configurations over the lattice of size $N+1$ with the restriction that $\sigma_1 = +1$. Furthermore, since $\overset{+}{u}$ begins with $+$, we know that $(\mathbf{P}(\overset{+}{u})\hspace{1mm} + \singlebar \overset{+}{u}\hspace{1mm} +)$ respects the periodic boundary conditions (given that $\overset{+}{u}$ begins with the plus, for the configuration $(\overset{+}{u}\hspace{1mm} +)$ (i.e. the first and last spins are up), with the periodic boundary conditions, the $K_{N+1}$ term must have the $+$ sign, as it does in $(\mathbf{P}(\overset{+}{u})\hspace{1mm} + )$). Similarly for the elements involving $\mathbf{A}(\overset{+}{u})$. Hence we have the desired form for the upper left entry. The form for the other three entries of the matrix is proved similarly.
\end{proof}

\bibliographystyle{plain}
\bibliography{bibliography}

\end{document}